\documentclass{IEEEtran}
\IEEEoverridecommandlockouts
\usepackage{amsmath,amssymb}
\usepackage{pdfsync}
\usepackage{subfigure}
\usepackage{multirow}
\usepackage{cite}
\usepackage{color}
\usepackage{pdfsync}
\usepackage{epsfig}
\usepackage{epstopdf}
\usepackage{times}

\newcommand{\eq}{\triangleq}
\newcommand{\field}[1]{\mathbb{#1}}
\newcommand{\R}{\field{R}}
\newcommand{\N}{\field{N}}

\newcommand{\Prob}{\mathbf{Pr}}
\newcommand{\E}{\mathbf{E}}
\newcommand{\hfs}{\hfill\ensuremath{\square}}
\DeclareMathOperator{\tr}{tr}
\DeclareMathOperator{\eigs}{eigs}
\DeclareMathOperator{\diag}{diag}

\newtheorem{ex}{Example}
\newtheorem{rem}{Remark}
\newtheorem{thm}{Theorem} 
\newtheorem{lem}{Lemma}

\title{Power Control and Coding Formulation for State Estimation with Wireless
  Sensors} 
\author{
  Daniel~E.~Quevedo*,~\IEEEmembership{Member,~IEEE,} Jan
  {\O}stergaard,~\IEEEmembership{Senior Member,~IEEE,} and
   Anders~Ahl\'en,~\IEEEmembership{Senior Member,~IEEE}
\thanks{Daniel Quevedo is
  with the School of Electrical Engineering \&
  Computer Science, The University of Newcastle, Australia; dquevedo@ieee.org. Jan  {\O}stergaard is with the Department of Electronic Systems, Aalborg
  University, Denmark; janoe@ieee.org.    
  Anders~Ahl\'en is with Signals and Systems, Uppsala University, Sweden;
  Anders.Ahlen@signal.uu.se.} 
\thanks{This research was supported under Australian Research Council's
  Discovery Projects funding scheme (DP0988601).}}

\begin{document}
%

\maketitle

\begin{abstract}
Technological advances have made wireless sensors cheap and reliable
enough to be brought into industrial use. A major challenge arises from the fact
that wireless channels  introduce random
packet dropouts. Power control and coding are key enabling technologies in wireless
communications to ensure efficient communications. In the present work, we
examine the role of power control and coding for Kalman filtering over
wireless correlated channels.  Two estimation architectures are
  considered: In the first, the 
sensors send their measurements directly to a single 
gateway. In the second scheme,  wireless relay nodes provide
additional
links. The gateway decides   on the  coding scheme and the  transmitter
power levels of the wireless nodes. The decision process is carried out on-line and adapts to varying
channel conditions in order to improve the trade-off between state
estimation accuracy and  energy expenditure. In combination with
predictive power control, we investigate the use of multiple-description coding,
zero-error coding and network coding and  provide sufficient conditions for the
expectation of the estimation error covariance matrix to be 
bounded. Numerical results 
suggest that the proposed method may lead to energy savings of around 50\%, when
compared to an alternative scheme, wherein transmission power levels and
bit-rates are governed by simple logic. In particular, zero-error  
coding is 
preferable at  time instances with high channel gains, whereas
multiple-description coding is superior for time instances with low
gains.  When channels between the sensors and the gateway are in
  deep fades, network coding improves estimation accuracy  significantly
  without sacrificing energy efficiency.
\end{abstract}

\section{Introduction}
\label{sec:intro}
Wireless sensors (WSs) have become  an important
alternative to wired sensors~\cite{ilymah04,ghakum03,willig08b}. WSs  
are equipped with a sensing 
component (to measure e.g., temperature), a 
processing device (to perform simple computations on the measured raw data), and
a communication device. WSs are cheap and reliable and offer several advantages,
such as, flexibility, low cost, and fast deployment. In addition, with  
WSs electrical contact problems are no longer an
issue. Furthermore, WSs and 
actuators can be placed where wires cannot go, or where power sockets are
unavailable. 

\par One major drawback of using WSs is that  wireless communication channels are subject to
fading and interference,  causing random packet errors\cite{goldsm05}.  The
time-variability of the fading channel can be alleviated by adjusting the power
levels and the transmitted packet lengths \cite{hantse99,gungus03}. To
keep packet error rates low, short 
packet lengths and high transmission power should be used. However, the use of
high transmission power is rarely an option, since in most applications WSs are
expected to be operational for several years without the replacement of
batteries; cf.,\cite{johbjo07}.  In addition,  short packets may require coarse
quantization which may lead to large 
quantization effects unless careful coding is
used\cite{jaynol84,covtho06}. It is safe to assume that, as in other wireless communication applications, power control
and coding will become key enabling technologies whenever WSs are
used. In particular, due to their wide applicability, including nonlinear
constrained MIMO systems (see,
e.g.,\cite{qiliu11a,versun11a,alegag11a,caiyan12a,queagu12a} for recent application studies), the
use of predictive control   
methods is worth investigating. 

\par In this work, we study two 
architectures having $M$ WSs and a single gateway 
(GW) for  Kalman-filter based state estimation of   linear time-invariant (LTI)
systems of the form:
\begin{equation}\label{eq:xk}
x(k+1) = Ax(k) + w(k), \quad k\in \N_0,
\end{equation}
where   $x(0) \in \mathbb{R}^n$ is zero-mean Gaussian distributed with covariance matrix $P_0$
and the driving noise process $\{w(k)\}_{k\in\N_0}$ is independent and identically
distributed (i.i.d.) zero-mean Gaussian distributed with
covariance matrix $Q$. 
The  measurement of sensor $m$, at time $k$, is given by
\begin{equation}\label{eq:yk}
y_m(k) = C_mx(k) + v_m(k), \quad m = \{1,\dotsc,M\}, 
\end{equation}
where $\{v_m(k)\}$ is i.i.d.\ zero-mean Gaussian measurement noise with
covariance matrix $R_m$.

\par  The first estimation architecture examined is depicted in
  Fig.~\ref{fig:scheme}  for the particular case of having $M= 2$ WSs. The measurements given by~(\ref{eq:yk}) are encoded and
transmitted at an appropriate power level over a fading channel
(generating random packet loss) to the GW. Received packets
are then used to  estimate $x(k)$  by means of a time-varying
Kalman filter (KF) which takes into account packet loss. As depicted in
Fig.~\ref{fig:scheme}, in addition to performing state estimation, the GW also
controls the power levels and the coding method (including bit-rate) used by the sensors at each
time. One of the main purposes of the present work is to show how predictive control
methods can be used for this purpose.  To keep
the sensors simple and energy efficient, the sensor nodes are not
allowed to communicate with each other. Thus, joint encoding of
the measurements taken by different sensors is not possible.
However, in the case where several measurements are received by
the GW, joint decoding is possible. By allowing separate
encoding to be followed by joint decoding, it is possible to take
advantage of distributed source coding techniques~\cite{slewol73,wynziv76}. In
the present work, we will focus on a 
particular distributed source coding technique known as zero-error
coding (ZEC)~\cite{koutun03}. In addition, to achieve robustness
in the presence of packet loss, we allow the sensors to use
multiple-description coding (MDC)~\cite{gamcov82}.

\begin{figure}[t]
  \centering
    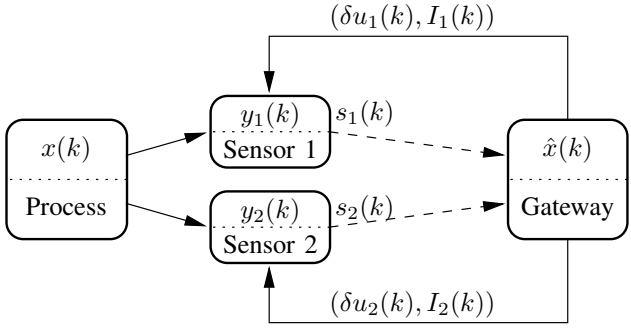
    \caption{State Estimation with $M=2$  wireless sensors.  The
        dashed lines denote fading channels which 
        introduce  
      random transmission errors. The gateway performs state estimation. It also
      controls the power level updates, $\delta u_m(k)$, and the coding method,
      as described in the codebook index $I_m(k)$.}
    \label{fig:scheme}
  \end{figure}

\par  In the second estimation architecture studied, the
  incorporation of $L$ relay
  nodes allows for additional communication links, see
  Fig.~\ref{fig:scheme_relay}. 
 Here, the measurements in~(\ref{eq:yk}) are quantized (with a uniform 
quantizer)  and
transmitted at an appropriate power level over  fading channels
 to the GW and relays. The latter perform
network coding and forward processed sensor measurements whenever appropriate to the GW.
To avoid interference between nodes, the communication channel is accessed in a TDMA
fashion with a pre-designed protocol.
 At the GW, received packets from the sensors and relays
are then used to  estimate $x(k)$ via Kalman filtering. For this second
architecture, the sensors do not perform MDC or ZEC. Thus, the codebook indices
$I_m(k)$ amount to the bit-rates to be used by the sensors.

 \par The main contribution of the present work is to investigate the role of
 dynamic power
control and   coding for state estimation with WSs through use of
 predictive control. 
The objective of the controller is to  counteract channel
  variability and to trade-off battery use for estimation accuracy.
It is located at the GW and decides upon the transmission power
level and coding scheme to 
be used by each node.
Our results indicate that it is advantageous that power levels
approximately invert channel gains provided sufficient power is available and that  MDC be
used at the
sensors when the channel conditions are poor.  
When good channel conditions are expected, it pays off to use ZEC across the sensors.
  If relays are available, then it turns out that, when channels
  between sensors and the GW are subject to severe fading, the use of network
  coding will improve the estimation performance significantly without
  increasing energy expenditure compared to the case with no relays. Hence,
  network coding is an attractive alternative/complement to MDC.

\begin{figure}[t]
  \centering
    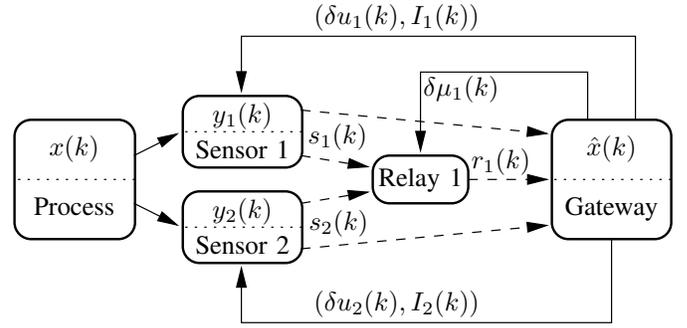
   \caption{State Estimation with two wireless sensors and $L=1$
        relay 
      node. Here the GW calculates $\hat{x}(k)$, the power level updates of the
      sensors and relays, and the bit-rate, $I_m(k)$, of each sensor node.}
    \label{fig:scheme_relay}
\end{figure}

\par  The present work extends our  recent work documented in
\cite{queahl08a,queahl09,queahl10,ostque09,ostque10}. The papers
\cite{queahl08a,queahl09} introduced the idea of using predictive control of WS
power levels for dynamic state estimation and control
applications. In\cite{queahl10}, the combination 
of power control and ZEC 
was considered, our  conference contribution \cite{ostque10}
  examined network coding for architectures with relays, whereas
in\cite{ostque09}, the combination of power control 
and MDC was considered. The results reported in
\cite{queahl10,ostque09,ostque10} 
indicate  that with simple coding techniques, significant improvements over the
uncoded case \cite{queahl08a} can be achieved. This motivates the present paper,
which combines power control, with either ZEC  and MDC, or network coding. 


\paragraph*{Notation}
\label{sec:notation}
The trace of a matrix $A$ is denoted by $\tr A$, and its spectral norm  by
$||A||\eq \sqrt{\max \eigs (A^TA)}$, where $\eigs (A^TA)$ are the eigenvalues
of $A^TA$ and the superscript $T$ refers to transposition. 
%
The Euclidean norm of a vector $x$ is denoted $|x|$;
$\Prob\{\cdot\}$ refers to probability, and $\E\{\cdot\}$  to expectation.
Discrete entropy  is denoted $H(\cdot)$; for   differential entropy we use
 $h(\cdot)$. 

\section{Coding Aspects}
In this section, we revise some basic aspects
on  source coding. 
Throughout this work, we will use  standard
high-resolution source coding results; see, e.g.,~\cite{gergra92}.

\subsection{Scalar Quantization, Entropy Coding, and High-Resolution Source Coding}\label{sec:indp_coding}
Each sensor $m$ encodes its measurement $y_m(k)\in \mathbb{R}$ into a quantized version $\hat{y}_m(k)$, which is further represented by a sequence of bits $s_m(k)$ to be transmitted over the channel, see Fig.~\ref{fig:coder}. The average bit-rate of $s_m(k)$ is denoted $b_m(k)$. The encoder consists of a  (time-varying) uniform scalar quantizer $\mathcal{Q}_m$ having step-size $\Delta_m(k)$, which is followed by an entropy encoder $\mathcal{E}_m$.  A scalar uniform quantizer can be efficiently implemented by simply  scaling $y_m(k)$ by $\Delta_m(k)$ followed by rounding,
i.e., by forming $\lfloor y_m(k)/\Delta_m(k)\rceil$ where $\lfloor\cdot\rceil$ denotes rounding to the nearest integer. 
If the GW receives $s_m(k)$, it reconstructs $\hat{y}_m(k)$ by simply applying the inverse
scaling,
\begin{equation*}
\hat{y}_m(k)=  \lfloor y_m(k)/\Delta_m(k)\rceil \Delta_m(k).
\end{equation*}

Under high-resolution assumptions, the bit-rate is given by\footnote{The approximation becomes exact in the limit as the distortion tends to zero~\cite{gergra92}. However, it is also known that these high-resolution results
 are approximately true even at rates as low as 2 bit/dimension; cf.~\cite{goyal00}.}
\begin{equation}
  \label{eq:rate}
  b_m(k)\approx H(\hat{y}_m(k)) \approx h(y_m(k)) - \log_2(\Delta_m(k)),
\end{equation}
where $\Delta_m(k)$ denotes the step-size and the expected distortion $D_m(k)$ satisfies
\begin{equation}
  \label{eq:9}
  \begin{split}
    D_m(k)&\eq \E\big\{|y_m(k)-\hat{y}_m(k)|^2\,\big|\,b_m(k)=b\big\} \\
    &\quad \approx
    \frac{1}{12}2^{-2(b-h(y_m(k))}.
  \end{split}
\end{equation}
In the sequel, we assume that $\{y_m(k)\}_{k\in\N_0}$ is a
stationary process with $y_m(k)$ being zero-mean Gaussian with
variance $\sigma_{y_m}^2$.

\begin{figure}[t!]
  \centering
    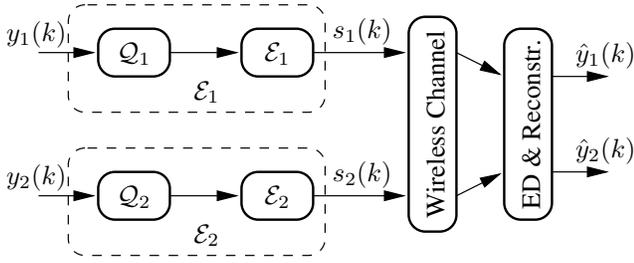
    \caption{Coding with $M=2$ WSs. Measurements $y_1(k)$ and
      $y_2(k)$ 
      are quantized, entropy coded and transmitted over  fading  
 channels. At the receiver, entropy decoding (ED) and reconstruction 
 yields $\hat{y}_1(k)$ and $\hat{y}_2(k)$.}
    \label{fig:coder}
\end{figure}

The entropy coder consists of a codebook\footnote{Entropy coding can be done by a simple table-lookup since the rounding (quantization) operation  directly gives the index of the codeword in the table.},
which due to memory considerations cannot be arbitrarily large.
In practice, we choose the size of the codebooks so that the probability of falling outside the support of the
entropy coder is very small and the impact of the outliers on the distortion is  negligible. Since   a codebook is needed for every possible  
$\Delta_m(k)$,  it is necessary  to discretize the alphabet of $\Delta_m(k)$ or,
equivalently, to discretize
the set of possible bit-rates $b_m(k)$.
In the following, we will assume that an appropriate discretization for a given
system $(A,C)$ is
found offline through, e.g., computer simulations. Thus, the controller uses the
constraint 
\begin{equation}
  \label{eq:4}
  b_m(k) \in\mathcal{B}_m,\quad  \forall m\in\{1,2\dots, M\}
\end{equation}
for given finite sets $\mathcal{B}_m\subset (0,\infty)$.
In particular, in our results documented in Section~\ref{sec:simulations}, we
confined the bit-rates 
$b_m(k)$ to the sets  $\mathcal{B}_m=\{3,\dotsc,8\}$.

\subsection{Zero-Error Coding}
\label{sec:zec}
As  mentioned in the introduction, the $M$ WSs are separated and cannot communicate with each other.
Encoding of the measurements can therefore not be done jointly.
However, the GW sees all the received measurements and thereby can  perform centralized joint decoding.
Thus, we are facing the \emph{distributed source coding} problem, i.e.,
separate encoding of $M$ correlated variables followed by joint
decoding~\cite{slewol73,wynziv76,covtho06}. In this work, the GW will, at times, 
command the WSs to adopt a
distributed source coding technique, known as zero-error coding (ZEC) \cite{koutun03}.
With ZEC, the measurements are quantized independently using the same scalar quantizers as previously
designed for the case of independent coding.
The only change is with regard to the entropy coder: rather than employing
independent entropy coding on the quantized measurements, with ZEC the WSs use dependent
entropy coders. More specifically,
they   adopt an asymmetric strategy, where one \emph{dominant} sensor, say sensor
$m^\star$, performs independent coding, 
i.e., independent scalar quantization followed by independent entropy coding. Hereafter,
another sensor, say sensor $m$, performs independent scalar quantization
followed by entropy coding  with respect to the entropy code of sensor $m$. With
this strategy, if the 
GW receives both $s_{m^\star}(k)$ and $s_m(k)$, then  it is possible to reconstruct $\hat{y}_{m^\star}(k)$ and
$\hat{y}_m(k)$. If only $s_{m^\star}(k)$ is received, then the GW can still
obtain $\hat{y}_{m^\star}(k)$, but of 
course not $\hat{y}_m(k)$. However, if only $s_m(k)$ is received, then the
GW cannot reconstruct neither 
$\hat{y}_{m^\star}(k)$ nor $\hat{y}_m(k)$.

\subsection{Multiple-Description Coding}
\label{sec:mdc}

The idea behind MDC is to create separate descriptions, which are individually
capable of
reproducing a source to a specified accuracy and, when combined, are able to
refine each other \cite{gamcov82}.  For that purpose, when using MDC, the source vector $y_m(k)$ is mapped to $J_m(k)$
descriptions
\begin{equation*}
s_m^i(k),\quad i\in\{1,\dots,J_m(k)\},
\end{equation*}
which are independently entropy coded and transmitted separately
to the GW. 

\par In this work, we will consider MDC based on index-assignments
 and lattice vector quantization~\cite{vaisha93,ostjen06,ostque09}.
 We will assume that for any $y_m(k)$, the packet-loss
probabilities for the $J_m(k)$
descriptions are i.i.d.\ and equal. Furthermore, we will focus on the symmetric
situation where the bit-rates  of each   description formed
at the $m$th sensor are equal, given by
$b_m(k)/J_m(k)$, and where the distortion observed at the
GW depends only upon the number of received descriptions
and not on which descriptions are received.

\subsection{XOR-based Network Coding}
\label{sec:proc-at-relays} 
In the second estimation architecture
under study relays are used to enhance estimation
performance. In this setup, sensor
data is sent by using simple independent coding, as described in
Section~\ref{sec:indp_coding}. The relays act
as intermediate network nodes and are able 
to perform simple XOR-based network coding on the data\cite{ahlcai00}. 
As
illustrated in 
Fig.~\ref{fig:scheme_relay}, the relay nodes are overhearing broadcast
communication 
from the 
sensors to the GW, and are therefore able to aid the GW with
additional information about the sensors' data. In
particular, the relays will XOR the incoming data at a bit level, i.e., without
decoding\cite{frabou06}.  Here one simply zero pads the
shortest symbols in order to make them all of equal length \cite{ostque10}.

\begin{ex}
Consider the scheme in Fig.~\ref{fig:scheme_relay} and
assume that the GW has received either only  ${s}_1(k)$ or only
${s}_2(k)$. If the relay  receives both ${s}_1(k)$  and
${s}_2(k)$, then it transmits
\begin{equation}
  \label{eq:4b}
 r_1(k)=s_1(k)\oplus s_2(k)
\end{equation}
 to the GW. If $r_1(k)$ is successfully received, then the  GW is  able to
recover both ${s}_1(k)$ and $s_2(k)$ and thereby reconstruct both values
$\hat{y}_1(k)$ and $\hat{y}_2(k)$ by use of $r(k)$ and its own message 
$s_1(k)$ or $s_2(k)$, see Table~\ref{tab:networkcoding}. \hfs
\end{ex}

\begin{table}[t]
  \begin{center}
\begin{tabular}{c|c}
 {\color{black}Data successfully received} &  {\color{black}Values reconstructed}\\ \hline
{\color{black}$s_1(k)$,  $s_2(k)$, $r_1(k)$} & {\color{black}$\hat{y}_1(k)$,
  $\hat{y}_2(k)$}
\\ \hline
{\color{black}$s_1(k)$,  $s_2(k)$}  & {\color{black}$\hat{y}_1(k)$,
  $\hat{y}_2(k)$ } \\ \hline 
{\color{black}$s_1(k)$,   $r_1(k)$} &{\color{black} $\hat{y}_1(k)$, $\hat{y}_2(k)$}  \\ \hline
{\color{black}  $s_2(k)$, $r_1(k)$} &{\color{black} $\hat{y}_1(k)$, $\hat{y}_2(k)$}  \\ \hline
{\color{black}$s_1(k)$ } & {\color{black}$\hat{y}_1(k)$}  \\ \hline
{\color{black}$s_2(k)$} & {\color{black}$\hat{y}_2(k)$ } \\ \hline
{\color{black} $r_1(k)$ }& {\color{black} none} \\ \hline 
{\color{black}none} & {\color{black} none} \\ \hline
\end{tabular}
\caption{Reconstructed values at the GW when using the estimation architecture in Fig.~\ref{fig:scheme_relay}
  with network coding as described in Section~\ref{sec:proc-at-relays}.} 
\label{tab:networkcoding} 
\end{center}
\vspace{-8mm}
\end{table}

\subsection{Key properties and complexity issues of the proposed coding schemes}
\subsubsection{Independent coding}
This is the simplest of the proposed architectures and is furthermore a fundamental part of zero-error coding as well as XOR-based network coding. Since the set of possible bit-rates $\mathcal{B}_m$ for the $m$th sensor is discrete (but actually not limited to integer valued elements due to entropy coding), the optimization over bit-rates is non-linear and non-convex. Fortunately, the cardinality of $\mathcal{B}_m$ can usually be chosen small in practice, and we therefore simply let the GW perform a brute-force search over all possible candidate bit-rates. 

\subsubsection{Zero-error coding}
The advantage of ZEC over independent coding is that one can reduce the rate of any given
entropy coder by making it dependent upon another entropy coder without increasing the complexity at the sensor nodes. 
The complexity at the GW is, however, increased, since the GW has to decide upon whether ZEC should be used or not, see Section \ref{sec:search}. 
If the channels from sensors ${m^\star}$ and $m$ to the GW are both reliable, or if at least
one of them is, then it is beneficial to exploit ZEC as is also evident from the simulations in Section \ref{sec:sim_zec}. 
It is worth emphasizing that a reduction of the number of bits to transmit, immediately translates into an energy reduction for a fixed transmission power. 


\subsubsection{Multiple-description coding}
When the channels are unreliable and causing packet losses, it is advantageous to use MDC and thereby sent multiple packets, see e.g., Fig. \ref{fig:mdc_vs_sdc}, which illustrates the reconstruction accuracy\footnote{The accuracy is measured before the GW applies its KF.} due to using MDC as a function of the channel quality. 
With the chosen approach to MDC, which is based upon index-assignments \cite{ostjen06} (i.e., table-lookups), the complexity at the sensor nodes is not increased over that of independent coding. Moreover, since closed-form solutions for the best choice of MDC parameters and codebooks exist, the complexity at the GW is only slightly increased. 
The bit-rates of the individual packets are generally smaller than those used for the single packet case. Moreover, the transmission power can often be reduced when using MDC, since it is more likely that at least one small packet out of several packets is received than one particular large packet is received. 

\subsubsection{XOR-based network coding}
If there is a relay available (e.g., one of the sensor nodes could act as a
relay node), which overhears broadcast messages, then it is advantageous to
exploit, e.g., simple XOR-based network coding whenever a subset of the channels
are experiencing fading. Since the individual sensors simply perform independent coding, their complexity is not increased. The complexity at the relay is determined by the XOR operations, which can be efficiently executed on most hardware architectures. Due to help of the relay, the individual sensors can reduce their transmission powers, which in turn saves energy.

\section{Transmission Effects and Power Issues}
\label{sec:transmission-effects}
 We will model transmission effects  by introducing 
the binary stochastic arrival processes
\begin{equation*}
  \gamma_m^i(k) =
  \begin{cases}
  1&\text{if $s_m^i(k)$ arrives error-free at time $k$, when}\\
  &\text{transmitted directly from sensor $m$ to the GW,}\\
  0& \text{otherwise.}
  \end{cases}
\end{equation*}
 Transmission effects when using  the estimation architecture
  with relays and where no MDC or ZEC is used, see Fig.~\ref{fig:scheme_relay},
  are  modeled in a similar  
  manner.  {\color{black} Here, we introduce the  
  binary stochastic arrival processes
$\zeta^\ell_m=\{\zeta^\ell_m(k)\}_{k\in \N_0}$ and 
${\tilde\gamma}_\ell=\{\tilde\gamma_\ell(k)\}_{k\in \N_0}$,  see
Fig.~\ref{fig:quantities}} and where
 \begin{equation*}
   \begin{split}
  \zeta_m^\ell(k) &=
  \begin{cases}
    1&\text{if $s^1_m(k)$ arrives error-free at time $k$, when}\\
    &\text{transmitted  from sensor $m$ to the $\ell$-th relay,}\\
    0& \text{otherwise,}
  \end{cases}\\
  \tilde\gamma_\ell(k) &=
  \begin{cases}
    1&\text{if $r_\ell(k)$ arrives error-free at time $k$ at the GW,}\\
    0& \text{otherwise.}
  \end{cases}
   \end{split}
 \end{equation*}

\subsection{Channel Power Gains}
\label{sec:channel-power-gains}
In the sequel,   we denote by
  $g_m(k)$  the complex channel gain (at time $k$) between the sensor $m$ and the
  GW,  by  $g^\ell_m(k)$ the channel gain between the $m$-th
    sensor and the $\ell$-th 
  relay, and by $\tilde{g}_\ell(k)$ the channel gain between the $\ell$-th relay
  and the GW, see Fig.~\ref{fig:quantities} 
 The transmission power used
by the  radio power amplifier of the $m$-th sensor is denoted $u_m(k)$,  
  whereas that of the $\ell$-th relay is $\mu_\ell(k)$. If we assume that the
packet length is equal to the 
 bit-rate, and that the bit errors are independent of one another at a given
 time $k$, then the conditional success
probabilities  
\begin{equation}
 \label{eq:5}
  \begin{split}
    \lambda_m^i(k) &\eq \Prob\{ \gamma_m^i(k) = 1\,|\,u_m(k),
    g_m(k),b_m(k),J_m(k)
    \}\\ &\quad= \big(1- \beta\big(u_m(k)
    |g_m(k)|^2\big)\big)^{\frac{b_m(k)}{J_m(k)}}\\
    {\color{black} \rho_m^\ell(k)} &\eq {\color{black}  \Prob\{ \zeta_m^\ell(k) =
      1\,|\,u_m(k), g^\ell_m(k),b_m(k)  \}}\\ &\quad= \big(1- \beta \big(u_m(k)
    |g^\ell_m(k)|^2\big)\big)^{b_m(k)}\\
    {\color{black} \tilde\lambda_\ell(k)} &\eq {\color{black}  \Prob\{
      \tilde\gamma_\ell(k) = 1\,|\,\mu_\ell(k), \tilde{g}_\ell(k),\tilde
      b_\ell(k) \}}\\ &\quad= \big(1- \beta \big(\mu_\ell(k)
    |\tilde g_\ell(k)|^2\big)\big)^{\tilde b_\ell(k)},
  \end{split}
\end{equation}
for $u_m(k),\mu_\ell(k)>0$ and where $\tilde b_\ell(k)$ is the largest packet length received by the
$\ell$-th relay at time $k$, see
Section~\ref{sec:proc-at-relays}. In~(\ref{eq:5}), the function
$\beta(\cdot)\colon [0,\infty ) \rightarrow [0,1]$ is the bit-error 
rate (BER). It is
monotonically decreasing function and depends on the
modulation scheme employed. 


\begin{figure}[t]
  \centering
    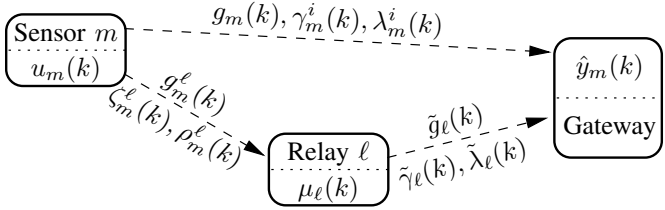
    \caption{{\color{black} Transmission power levels, channel gains,
        transmission outcomes and conditional success 
      probabilities for the estimation architecture with relays and
      sensors. Only one sensor and relay are shown.}}
    \label{fig:quantities}
  \end{figure}

\par It follows from~\eqref{eq:5}, that one
can improve transmission reliability  and, thus, state estimation
accuracy for  given channel gains, by
transmitting shorter packets and/or  by  increasing the
power levels used by the transmitters. However, as we have seen in
Section~\ref{sec:mdc}, smaller values of packet lengths $b_m(k)$
will lead to larger quantization distortion.
Furthermore, the success
probabilities are affected by the channel power gains of the
different channels. Hence, the statistical properties of the channels
will have an impact on the transmit power devised by the controller.

\par We shall assume a block (flat) fading channel model where the complex
  gains $g_m(k), g^\ell_m(k), \tilde g_\ell(k)$ are all constant over
the duration of a packet, and fading
between packets. In many wireless sensor network applications 
the complex channel gains are assumed to be i.i.d.\ which would make
sense if measurements are transmitted rarely as compared to the 
fading speed. Here we shall, however, adopt a more general model where  complex channel
gains are correlated. 
In particular, when the fading channel taps are subject to Rayleigh fading, it
is convenient to  adopt
a first order Markov model\cite{wancha96} 
of the form\footnote{The use of higher order models, as adopted
  for example in \cite{linahl02,badbea05} is straightforward.}
\begin{equation}
\label{eq:43b}
\begin{split}
g_m(k) &= a_{g_m} g_m(k-1) + e_{g_m}(k)\\
g^\ell_m(k) & = a_{g^\ell_m} g^\ell_m(k-1) + e_{g^\ell_m}(k),\\
\tilde g_\ell(k) &= a_{\tilde g_\ell}\tilde g_\ell(k-1) +e_{\tilde
  g_\ell}(k)
\end{split}
\end{equation}
where $a_{g_m}, a_{g^\ell_m},  a_{\tilde g_\ell}$ determines the
amount of correlation and where $e_{g_m}(k), e_{g^\ell_m}(k), e_{\tilde
  g_\ell}(k) $ are mutually independent zero mean circular symmetric complex Gaussian white noises with
appropriate covariances. 
To reduce   complexity, the GW discretizes the
instantaneous fading gains of these channels into $N$ intervals $[\Gamma_{n},
\Gamma_{n+1}]$,  $n \in \{0,\dotsc, N-1\}$, $\Gamma_0 = 0$, $\Gamma_N=\infty$,
and adopts a homogeneous finite state Markov chain
(FSMC) model with  associated states $\varsigma_n$, $n\in\{0,\dotsc, N-1\}$~\cite{wanmoa95,golvar96}. 
 The probability that the channel gain
switches from state $\varsigma_n$ to state $\varsigma_j$ within a single time
step is denoted 
$p_{n,j}$. 
We assume that channel states switch only between neighbors, thus,
$p_{n,j} = 0$,  for all $|n-j| > 1$.


\subsection{Energy Use}
\label{sec:energy-use}
When using WSs it is of fundamental importance to save energy. We thus have to
find a suitable balance between the transmit power used and the
estimation accuracy obtained. The energy used by each
sensor $m\in\{1,\dots,M\}$ to transmit $s_m(k)$ can be quantified via
\begin{equation*}
  E_m(b_m(k)u_m(k))
 \eq
 \begin{cases}
   \dfrac{b_m(k)u_m(k)}{r} + E_{\text{P}} &
   \text{if  $u_m(k)>0$,}\\
   0 & \text{if  $u_m(k)=0$,}
 \end{cases}
\end{equation*}
where 
$E_{\text{P}}$ denotes the processing cost, i.e., the energy needed for
wake-up,  circuitry and sensing, and $r$ is the channel bit-rate.
\par Due to physical limitations of the radio power amplifiers, the power levels are
constrained, for given saturation levels $\{u_m^\text{max}\}$, according to:
\begin{equation}
  \label{eq:15gg}
  0\leq u_m(k) \leq u_m^\text{max}, \quad \forall k \in\N_0,\quad \forall
  m\in\{1,2,\dots,M\}.
\end{equation}

\par The energy
consumption of the relays can be quantified similarly by introducing energy functions
$\tilde{E}_\ell(\mu_\ell(k)\tilde b_\ell(k))$. For simplicity, we will focus on relays operating in
 on-off mode, with  pre-determined  transmission power levels $\{\mu_\ell\}$, thus
\begin{equation}
  \label{eq:2}
  \mu_\ell(k) \in\{0,\mu^\text{max}_\ell\}, \quad \forall \ell \in\{1,2,\dots,L\}
\end{equation}
{\color{black} The relays
  transmit only if  the controller has 
  assigned  $\mu_\ell(k) = \mu^\text{max}_\ell$ and sensor 
data which is needed to perform network coding has been successfully received, see
Section~\ref{sec:proc-at-relays}. }

\section{Kalman Filtering with Multiple Intermittent Sensor Links and Coding}
\label{sec:kalm-filt-with}
State estimation is performed at the GW, which also governs the code-book
selection of the sensors, see Figs.~\ref{fig:scheme}  and
\ref{fig:scheme_relay}. 
We will assume that the GW knows, whether packets
received from
the sensors contain errors or not.
\footnote{This can be handled by the use of a
  simple cyclic redundancy check.} Thus, at  time $k$,  past and present
realizations of all
transmission processes
$ \{ \gamma_m^i(k-t)\}_{ t \in\N_0}$,
    $i \in \{0,\dots, J_m(k)-1\}$, $m\in\{1,\dots,M\}$,
and, {\color{black} in case of the estimation architecture with relays, the
  transmission outcome processes associated to links
  from relays to the GW, namely,}
$ \{ \tilde\gamma_\ell(k-t)\}_{t \in\N_0}$,
    $\ell \in \{1,\dots, L\}$,
are available at the GW to form the state estimate $\hat{x}(k)$. 

  \begin{table}[t]
    \begin{center}
      \begin{tabular}{c|c|c|c|c|c|c}
        $\gamma_1^{1}(k)$ & $\gamma_2^{1}(k)$ & $\zeta_1^{1}(k)$ &
        $\zeta_2^{1}(k)$ & $\tilde\gamma_1(k)$ &  $\theta_1(k)$ & $\theta_2(k)$ \vspace{0.2mm}\\  \hline
        1&1&1&1&1&1&1\\ \hline\hline
        1&1&1&0&0&1&1\\ \hline
        1&1&0&1&0&1&1\\ \hline
        1&1&0&0&0&1&1\\ \hline\hline
        1&0&1&1&1&1&1\\ \hline\hline
        0&1&1&1&1&1&1\\ \hline\hline
        1&0&1&0&0&1&0\\ \hline
        1&0 & 0&1&0 & 1&0\\ \hline
        1&0 & 0&0&0 & 1&0\\ \hline\hline
        0&1&1&0&0&0&1\\ \hline
        0&1 & 0&1&0 & 0&1\\ \hline
        0&1 & 0&0&0 & 0&1\\ \hline\hline
        0&0 & 1&1&1 & 0&0\\ \hline\hline
        0&0 & 1&0&0 & 0&0\\ \hline
        0&0 & 0&1&0 & 0&0\\ \hline
        0&0 & 0&0&0 & 0&0\\ \hline
      \end{tabular}
      \caption{{\color{black} Reconstruction processes of the architecture in
          Fig.~\ref{fig:scheme_relay} with network coding as described in
          Section~\ref{sec:proc-at-relays}. Note that the relay transmits only
          if it has received both symbols $s_1(k)$ and $s_2(k)$.}}
      \label{tab:reconstruction}
    \end{center}
    \vspace{-10mm}
  \end{table}

\subsection{The Reconstruction Processes}
\label{sec:reconstr-proc}
To elucidate the situation
 we introduce the discrete reconstruction
processes $\{\theta_m(k)\}_{k\in\N_0}$, via
  \begin{equation*}
     \theta_m(k) =
    \begin{cases}
      { 1}&\text{if $\hat{y}_m(k)$ can be
          reconstructed at time $k$,}\\ 
       { 0}&  \text{otherwise.}
    \end{cases}
  \end{equation*}
Clearly these processes depend on the transmission outcomes, as
dictated by the coding schemes employed. More precisely, if the $m$-th sensor uses only independent coding and no
    network coding, then  
    $\theta_m(k)=\gamma^{1}_m(k)$. If, in the estimation architecture of
    Fig.~\ref{fig:scheme}, at time $k$ the sensor $m$ uses 
    MDC,
    then $\theta_m(k)=1$, if and only if at least one of the $J_m(k)$
    descriptions of $y_m(k)$ is
    successfully received at the GW and we have:
  \begin{equation*}
   \theta_m(k) = 1-\prod_{i\in\{1,2,\dots, J_m(k)\}} \big (1-\gamma_m^i(k)\big).
  \end{equation*}
If ZEC with dominant coder $m^{\star}$ is used, then
    \begin{equation}
  \label{eq:45b}
  \theta_m(k) = \begin{cases}
   \gamma_m(k)\gamma_{m^\star}(k) &
   \text{if  $u_m(k) u_{m^\star}\!(k)>0$,}\\
   0 & \text{if  $u_m(k)u_{m^\star}\!(k)=0$,}
 \end{cases} 
\end{equation}
 for all $m\in \{1,2,\dots,M\}$, see \cite[Eq.\ (17)]{queahl10}.  On
    the other hand, if in the setup with relays, see Fig.~\ref{fig:scheme_relay},
    network coding is used at time $k$, then $\theta_m(k)$ also depends upon the
    transmission outcomes involving the relays, namely, $\tilde{\gamma}_\ell(k)$ and
    $\zeta_m^\ell(k)$.
For example, for the case given in
Table~\ref{tab:networkcoding},  the processes $\theta_1(k)$ and  $\theta_2(k)$
are determined as per
Table~\ref{tab:reconstruction}.

\subsection{System Model}
\label{sec:signal-model}
  The GW knows which coding method was used at
  current time $k$
  and
  also has access to the $M$ reconstruction process realizations. Thus, for state
  estimation purposes, the overall system amounts 
  to sampling~(\ref{eq:xk})-(\ref{eq:yk}) only at the successful transmission
  instants of each sensor link. It is convenient to model the overall
  estimation architectures by
  introducing the discrete stochastic output matrix process $\{C(k)\}_{k\in\N_0}$ and
  associated measurements $\{y(k)\}_{k\in\N_0}$  as in:\footnote{We note that, if MDC is used, then 
 $\hat{y}_m(k)$ in~(\ref{eq:3}) denotes the
reconstruction of $y_m(k)$ based on the $0\leq j\leq J_m(k)$
received descriptions.  This value differs from $y_m(k)$ due to the
  measurement noise $v_m(k)$ and the quantization noise.}
\begin{equation}
  \label{eq:3}
  C(k)={\mathcal C}(\theta(k))\eq
  \begin{bmatrix}
    \theta_1(k) C_1 \\  \theta_2(k) C_2\\ \vdots \\
 \theta_M(k)C_M
\end{bmatrix}\quad,\;
   y(k) \eq
   \begin{bmatrix}
   \theta_1(k)\hat{y}_1(k)\\ \theta_2(k)\hat{y}_2(k) \\ \vdots \\
   \theta_M(k)\hat{y}_M(k)
   \end{bmatrix},
\end{equation}
where
\begin{equation*}
  \theta(k) \eq
  \begin{bmatrix}
    \theta_1(k)&\theta_2(k)& \dots & \theta_M(k)
  \end{bmatrix}.
\end{equation*}

\par The following time-varying KF gives the best linear
estimates of the system state in~(\ref{eq:xk}) given the information available
at the GW,  for both estimation architectures
under study: 
\begin{equation}
  \label{eq:KF}
  \begin{split}
    \hat{x}(k)&=\hat{x}(k|k-1) +K(k) \big(y(k) - C(k) \hat{x}(k|k-1)\big)\\
    \hat{x}(k + 1|k) &= A\hat{x}(k|k-1) + AK(k)
    \big(y(k) - C(k) \hat{x}(k|k-1)\big)\\
    P(k+1|k) 
    &= A (I-K(k)C(k)) P(k|k-1)A^T+Q
\end{split}
\end{equation}
where  the gain $K(k)$ and equivalent measurement noise covariance $R(k)$ are given by
\begin{equation}
  \label{eq:8}
  \begin{split}
 K(k)&\eq P(k | k-1) C(k)^T\big(C(k)P(k | k-1)C(k)^T +R(k) \big)^{-1}\\
R(k)&\eq \diag \big(R_1+D_1(k),\dots,R_M+D_M(k)\big).
\end{split}
\end{equation}
 Here, $D_m(k)$ are the distortions due to quantization, whereas $Q$ and $R_m$
 are the driving noise and measurement noise 
covariances, respectively. The recursion~(\ref{eq:KF}) is initialized with
$P(0)=P_0$ and $\hat{x}(-1)=0$, see~(\ref{eq:xk}). In the linear Gaussian case,
$P(k+1|k)$  corresponds to
the prior covariance of the estimation error. The posterior error covariance matrix
is then given by \cite{andmoo79}
\begin{equation}
  \label{eq:7}
  P(k|k) = (I-K(k)
  C(k))P(k|k-1).
\end{equation}

\begin{rem}
  Since $C(k)$ is known at the GW, the above Kalman
  filter uses all successfully reconstructed measurements to form the state
  estimate $\hat{x}(k)$. Those measurements where $\theta_m(k)=0$ are not taken
  into account.  This property is reflected in the expression for the filter
  gain $K(k)$, where the time-varying matrix $C(k)^T$ pre-multiplies the term
  $\big(C(k)P(k|k-1)C(k)^T +R(k) \big)^{-1}$. Thereby, $\hat{x}(k)$ is not updated
  based on those values, see~(\ref{eq:KF}).\hfs
\end{rem}

\section{On-line Design  of  Coding and Power Levels}
\label{sec:predc-power-contr}
We have seen that transmission power and bit-rate
design involves a trade-off between transmission error probabilities (and, thus,
state estimation accuracy)  and energy use.  We will next present a predictive
controller which optimizes
this trade-off.
To keep the sensors simple, the controller is located at the
GW.  For the first estimation architecture, the
controller output 
contains information on the 
power levels, and the codebooks to be used by the 
sensors.  In case of the second estimation architecture, the
  controller also updates the power levels used by the relays.

\subsection{Signaling}
\label{sec:bit rate-reduction}
To save signal processing energy  at the sensors {\color{black} and relays}, we
would like to limit power control
signaling as much as possible.
The command signal for each sensor $m$ will contain, in addition
to the
codebook index $I_m(k)$,
 a finitely quantized power increment, say $\delta
u_m(k)$. 
 Upon reception of
 $(\delta u_m
(k),I_m(k))$, the 
sensor chooses the codebook $I_m(k)$ (to be used for
encoding $y_m(k)$),  and
reconstructs the power level to be used by its radio power
amplifier by simply setting
$ u_m(k) = u_m(k -1)+\delta u_m(k)$.
For the second estimation architecture, see
  Fig.~\ref{fig:scheme_relay}, only one bit is needed to
  convey each of the relay power levels 
  $\mu_\ell(k)$, see~(\ref{eq:2}).

\subsection{Cost Function}
\label{sec:cost-function}
To trade  energy consumption for estimation cost, at
each  instant $k$, the 
controller determines the power level increments $\delta u_m(k+1)$ and codebook
indices $I_m(k+1)$ and, for the
estimation architecture in Fig.~\ref{fig:scheme_relay} also the relay power
levels $\mu_\ell(k+1)$,   by minimizing the  cost function
\begin{equation}
  \label{eq:Vk}
  \begin{split}
    V&(S(k+1))\\
    &\eq\E\big\{
    \tr{P}(k+1|k+1)\,\big|\, g(k),P(k+1|k),S(k+1)\big\}\\
    &\qquad+\varrho V_E(S(k+1)),
  \end{split}
\end{equation}
where
\begin{equation}
  \label{eq:27}
  \begin{split}
    S(k+1) &= \{I_m(k+1),\delta u_m(k+1),\mu_\ell(k+1)\}\\ &\qquad
    m\in\{1,2,\dots,M\},\;\ell \in\{1,2,\dots,L\} .
  \end{split}
\end{equation}
The first term in~(\ref{eq:Vk}),  quantifies the estimation quality. 
The conditional expectation in~(\ref{eq:Vk}) is taken for
given $P(k+1|k)$, and channel gains 
$g(k)$,  which are available at the GW at time $k$.  Averaging is with respect to
 the set 
of
possible reconstructions and transmission outcomes due to
receiving different 
subsets of descriptions for 
each sensor. The probability distribution of these sets depends upon the decision
variables, i.e., the
bit-rates,  coding schemes and power levels.

 \par The second term in $V(S(k+1))$, 
 quantifies the energy use at the next time
step.\footnote{If desired, a term which penalizes the size of the 
power control signal $\delta u_m(k+1)$ (or  of $S(k+1)$), which is transmitted
from the gateway to the sensors can be readily included in the cost function.} Thus, $\varrho\geq 0$ is a tuning parameter which allows the designer to
trade-off estimation accuracy for energy use  at the sensors and relays. For  the architecture without
relays the latter is given by:
\begin{equation*}
  V_E(S(k+1))= V^{S}_E(S(k+1))\eq
  \sum_{m=1}^{M}E_m(b_m(k+1)u_m(k+1)),
\end{equation*}
{\color{black} whereas, for the estimation architecture with relays,}
\begin{equation*}
    V_E(S(k+1))=V^{S}_E(S(k+1))+\sum_{\ell=1}^{L}\tilde{E}_\ell(\tilde
b_\ell(k+1)\mu_\ell(k+1)).
\end{equation*}
On-line optimization of the cost function in~(\ref{eq:Vk}), gives rise to the
desired power levels and coding strategy,
\begin{equation}
  \label{eq:1}
S(k+1)^{\text{opt}}=\arg\min_{S(k+1)\in\, \mathbb{S}_{k+1}}  V(S(k+1)),
\end{equation}
where the finite set $\mathbb{S}_{k+1}$ represents the constraints on the decision variables
$I_m(k+1)$, $\delta u_m(k+1)$, and $\mu_\ell(k+1)$. In particular, power levels
and their increments are finite-set constrained such that \eqref{eq:15gg} and
\eqref{eq:2} are satisfied. The resulting controller is non-linear, constrained,
stochastic and adapts to  channel
conditions and current estimation quality, i.e., we have
\begin{equation}
  \label{eq:21}
  S(k+1)^{\text{opt}}= \kappa \big( g(k),P(k+1|k)\big)
\end{equation}
for some mapping $\kappa(\cdot,\cdot)$. It is worth noting that, in general, due
to the cost function being nonlinear and  constraints finite, no closed form
expression for $\kappa(\cdot,\cdot)$ exists. Instead, the optimization
in~(\ref{eq:21}) needs to be performed numerically. To  keep computations low, the cost
in~(\ref{eq:Vk}) looks  ahead at only
 one step.\footnote{Multi-step extensions
 can be easily formulated, following as in\cite{queahl10}} In Section~\ref{sec:controller}, we
discuss some computational issues.

\subsection{Performance Bound}
\label{sec:stochastic-stability}
The effect on packet drops on Kalman filter stability and performance has received
considerable attention in the recent
literature\cite{sinsch04,liugol04,huadey07}. Despite the 
fundamental importance of power control and coding in wireless communications, it
is somewhat surprising that these techniques have received so little attention
in this context. In
fact, to the best of the authors' knowledge, the only works examining power
control for Kalman estimation with packet dropouts are our own
\cite{queahl08a,ostque10,queahl10,queahl12}.  
Theorem~\ref{thm:perf}, stated below, establishes sufficient conditions for the
expected value of the 
state estimation 
error covariance $P(k+1|k)$ to be exponentially bounded. Whilst the main
interest of the current work is on state estimation for stable systems,
Theorem~\ref{thm:perf} also applies to unstable systems,
extending our recent results  documented
in\cite{queahl11b,queahl12,queahl13a} to controller structures such as~(\ref{eq:21}) which
are allowed to use 
the covariance matrices $P(k+1|k)$.
\begin{thm}
\label{thm:perf}
Consider the stochastic process
\begin{equation*}
 \eta (\theta) \eq
\begin{cases}
  1&\text{if $\mathcal{C}(\theta)$ is full rank,}\\
  0& \text{otherwise}
  \end{cases}
\end{equation*} 
and define\footnote{It follows from the analysis in the appendix that this
  conditional probability is independent of time $k$.} 
\begin{equation*}
  \nu(P,g) \eq  \Prob\{\eta(\theta(k+1)) = 0\,|\, P(k+1|k)=P,g(k)=g\}.
\end{equation*}
Suppose that there exists a uniform bound $\rho\in[0,1)$ such that
\begin{equation}
\label{eq:13}
  \nu(P,g) \leq \frac{\rho}{\|A\|^2}, \quad
  \forall (P,g)\in\R^{n\times n} \times \Omega,
\end{equation}
where $\Omega$ is the support of the (complex) channel gains.
Then \begin{equation}
    \label{eq:40}
      \E\big\{ \|P(k|k-1)\|\big\} \leq \rho^k \tr P_0 +\frac{\varpi c+\tr Q}{1-\rho} (1-\rho^k),
      \quad \forall k\in      \N_0,
  \end{equation}
where
\begin{equation}
  \label{eq:30}
  \begin{split}
    c&\eq \max_{\theta \in \{0,1\}^M \colon \eta (\theta)=1}
    \|\mathcal{C}^\dagger(\theta)^T\mathcal{C}^\dagger(\theta)\|, \\
    \mathcal{C}^\dagger(\theta)&\eq
    (\mathcal{C}(\theta)^T\mathcal{C}(\theta))^{-1}\mathcal{C}(\theta)^T\\ 
    \varpi&\eq||A||^2  \Big((\pi e/{6})\sigma_{y_m}^2 \sum_{m=1}^M2^{-2\check{b}_m} +
    \sum_{m=1}^M R_m\Big)
  \end{split}
\end{equation}
and $\check{b}_m\eq \min\{\mathcal{B}_m\}$, see~(\ref{eq:4}).
\end{thm}
\begin{proof}
  See appendix.
\end{proof}

Theorem~\ref{thm:perf} establishes a bound on the expected value of
$\|P(k|k-1)\|$ which decays with exponential rate $\rho$.  The latter quantity
depends on the spectral norm of the system matrix $A$ and the conditional
probabilities of the observation matrix $\mathcal{C}(\theta)$ being
full-rank. Thus, our result allows one to infer upon estimation accuracy from
conditional channel dropout probabilities.


\section{Computational Issues}
\label{sec:controller}
In the estimation architecture considered, the sensors do 
  not need to carry out significant computations. In particular, 
optimizations  and 
  Kalman filter recursions are performed by the gateway, see
  Fig.~\ref{fig:scheme}.  In relation to Kalman filtering, it is necessary to invert the matrix $C(k)P(k | k-1)C(k)^T +R(k)$
given in~(\ref{eq:8}). This matrix is of size $M\times M$ and therefore scales with the number of sensors. Since it is positive semidefinite and symmetric, its inverse can be obtained with efficient algorithms based upon Cholesky factorizations and having complexity on the order of $\mathcal{O}(M^3)$. 
 
To analyze computational cost of the optimizations in~(\ref{eq:1}), we first note that, from~(\ref{eq:7}), we have 
$ {P}(k+1|k+1)= \big( I - K(k+1) C(k+1) \big)   P(k+1|k)$.
 Thus, to evaluate the cost function in~(\ref{eq:Vk}) and find the constrained optimizer $S(k+1)^{\text{opt}}$, the controller first uses
$P(k|k-1)$ and
$\theta(k)$ to calculate $C(k)$ and $P(k+1|k)$
using~(\ref{eq:3}),~(\ref{eq:KF}) and~(\ref{eq:8}). Clearly, $K(k+1)$ and $C(k+1)$
depend upon $P(k+1|k)$, $\theta (k+1)$ and $R(k+1)$. The latter quantity
depends upon the decision variable $S(k+1)$. On the other hand, the
reconstruction processes at time $k+1$ depend upon the transmission 
outcomes processes at time $k+1$, as per the mapping induced by the codebook choice; see, e.g., Tables~\ref{tab:networkcoding}
and~\ref{tab:reconstruction}. Thus, for the architecture in
Fig.~\ref{fig:scheme}, the conditional expectation in~(\ref{eq:Vk})
can be evaluated  by using the conditional probabilities
$\Prob\{\gamma^i_m(k+1)=1\,|\,g(k),P(k+1|k),S(k+1) \}$ as required. Note that, by using~(\ref{eq:5})
 and the law of total probability, one obtains that
\begin{equation*}
  \begin{split}
    \Prob&\{\gamma^i_m(k+1)=1\,|\,g(k),P(k+1|k),S(k+1)\}\\
    &=\Prob\{\gamma^i_m(k+1)=1\,|\,g(k),S(k+1) \}\\ 
    &=\E\big\{\Prob\{\gamma^i_m(k+1)=1\,|\,g_m(k+1),g(k),S(k+1) \}\big\}  \\
    &=\E\big\{\Prob\{\gamma^i_m(k+1)=1\,|\,g_m(k+1),S(k+1)
    \}\big\}\\ &=\E\big\{\lambda^i_m(k+1)\big\}, 
  \end{split}
\end{equation*}
where expectation is taken with respect the conditional channel gain distribution
$\Prob\{g_m(k+1)\,|\,g(k)\}$, see the FSMC model in Section~\ref{sec:channel-power-gains}.

\subsection{Two-stage search strategy}\label{sec:search}
 As outlined above, at every instant $k$, the proposed controller first finds the optimal
set of power value increments and bit-rates.  In the estimation
  architecture with relays, the sensors simply perform independent coding. In
  contrast, in the estimation architecture depicted in Fig.~\ref{fig:scheme}, each sensor can either do independent
coding, ZEC, or MDC. 
Moreover, in the case of MDC, for any given bit-rate, the controller also
needs to decide upon the number of descriptions $J_m(k+1)$ and the level of
redundancy between the $J_m(k+1)$ descriptions, see\cite{ostque09}.

\par To develop a simple but efficient method to select the coding scheme for
the estimation architecture without relays,
we propose  a two-stage search strategy.
In the first stage, the GW only evaluates whether independent coding or MDC
should be used at each sensor. In the second stage, if independent coding
is chosen for more than one sensor, then the GW further evaluates whether ZEC
should be used across these  sensors (or a subset of them). The second
stage uses exhaustive search.
 The first stage can be implemented very efficiently. Here it is important to
note that the energy
consumption at a given sensor is 
independent of the number of descriptions chosen.
Furthermore, under
high-resolution assumptions, it is reasonable to assume that the
quantization error  at any given sensor does not contain
significant information about current and past measurements of any
of the sensors.\footnote{Unless one utilizes substractively dithered
quantizers, the quantization error will generally not be
independent of the input signal. However, the quantization error
can be made uncorrelated with the input signal.} 
Thus,
from the KF point of view, $\hat{y}_m(k)$ amounts to a noisy version of
$y_m(k)$; the smaller the variance of this noise, the better the
estimate $\hat{x}(k)$. This motivates us to adopt a method where,
having chosen the optimal $(\delta_m(k),b_m(k))$ and, thus $u_m(k)$,
the controller selects the quantization scheme which results in the minimum
expected distortion on $y_m(k)$.\footnote{Our results  in
  Section~\ref{sec:simulations},  use $x(k)$ and rather than $y(k)$ and therefore take the effect of
  the Kalman filter into account when finding the coding scheme that yields the minimum expected distortion.}
Furthermore, given the bit-rate $b_m(k)$ and by considering
$\{\lambda_m^i : i=0,\dotsc, J_m(k) \}, \forall J_m(k)$ as weights, the  simple
method to find the optimal number of descriptions as well as the
optimal amount of redundancy between the descriptions proposed
in~\cite{ostjen06} will be used. The following example illustrates the procedure:

\begin{figure}[t!]
      \centering 
\includegraphics[width=\linewidth]{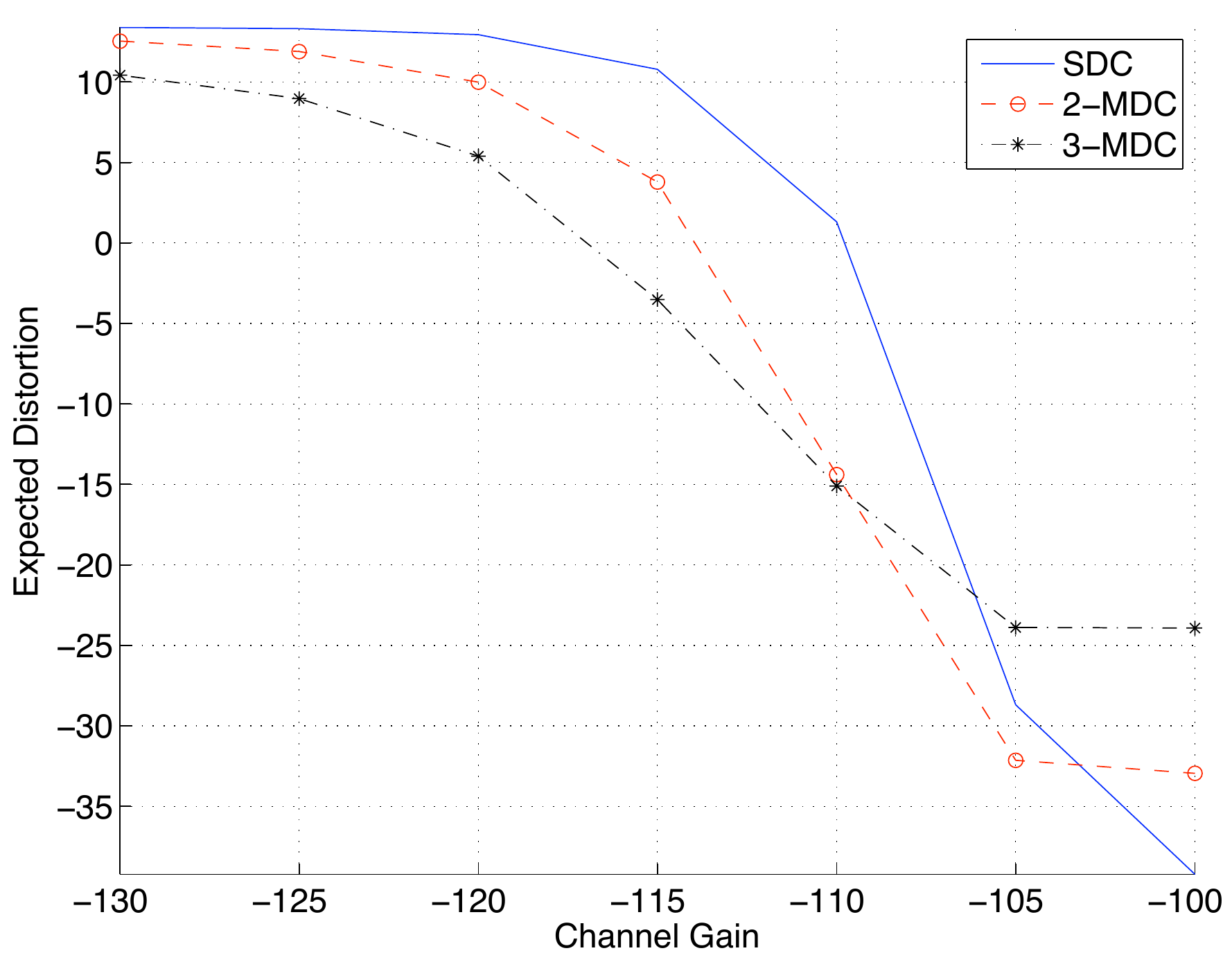}
\caption{MDC vs.\ SDC as a function of channel power gain, $|g_m|^2$. The total bit-rate is fixed at 9 bits/sample; $u_m=5\cdot 10^{-5}$.}
\label{fig:mdc_vs_sdc}
\end{figure}

\subsection{Example of the First Stage Search}
\label{sec:two-stage}
We consider a  second order linear time-invariant system with
  transfer function
  \begin{equation*}
    \label{eq:16}
    \frac{5.2978 (s+19.46)}{(s^2 + 0.05214s + 33.3)},
  \end{equation*}
which upon sampling with a sampling period of $0.1$ [s] can be written in the
form~(\ref{eq:xk}) with 
\begin{equation}\label{eq:A}
  A=
  \begin{bmatrix}
    1.6718 & -0.9948 \\
    1 & 0 \\
  \end{bmatrix}.
\end{equation}
{\color{black}The system poles are oscillatory, located at
  $0.8359 \pm 0.5441 \imath$. The driving noise covariance matrix is chosen as 
$Q=1/2I$, whereas } $P_0=0.3I$.
We study an estimation architecture  having no relays, $M=2$ WSs with noise
variances $R_1=R_2=1/100$. The individual observation matrices are $C_1=
[ 1\;0]$, and $C_2=
 [ 0\;1]$, thus, $c=1$, see~(\ref{eq:30}).

\par 
We analyze the first stage of the two-stage search strategy proposed in Section~\ref{sec:controller}. In other words, we show that the GW can decide up-front whether MDC or independent coding should be used for any given bit-rate.
We focus on a single sensor and fix the bit-rate $b_m(k)=9$ bits/sample and the power level $u_m(k)= 5 \cdot 10^{-5}$.
Fig.~\ref{fig:mdc_vs_sdc}, shows the expected distortion as a function of the
channel power gain, $|g_m|^2$ as observed by the GW before applying the KF. In
this plot, the expected distortion depends upon the coding scheme. 

\par In the case of independent coding (denoted SDC for
``single-description coding''), the encoded description $s_m(k)$ is
either received error-less or considered lost with probability
$\lambda_m(k)$ and $1-\lambda_m(k)$, respectively. Thus, the
expected distortion is given by $\lambda_{m}(k)D_m(k) +
(1-\lambda_m(k))\sigma_{y_m}^2(k),$ since, if $s_m(k)$  is lost, we have
$y_m(k)-\hat{y}_m(k)=y_m(k)$. 
 Similarly, with MDC the expected distortion is a weighted sum over the
distortion due to receiving subsets of descriptions. In Fig.~\ref{fig:mdc_vs_sdc},
2-MDC refers to the case of using 2 descriptions at the given
sensor and 3-MDC refers to the case with 3
descriptions.

\par The analysis suggests that   when the channel is in a
deep fade, it it better to use three descriptions. In the
mid-range of channel SNR, it is better to use two descriptions. When the channel
is very good, SDC is preferable. (Recall that when  MDC is used, the
total bit-rate is fixed at $b_m(k)=9$ bits/sample and evenly split
across the descriptions. Hence, in the case of 3 descriptions,
each description is encoded at 3 bits/sample. Since the packet
length (bit-rate) is reduced, it becomes more probable that a
description will be received without errors.)


\begin{figure}[t!]
      \centering 
\includegraphics[width=\linewidth]{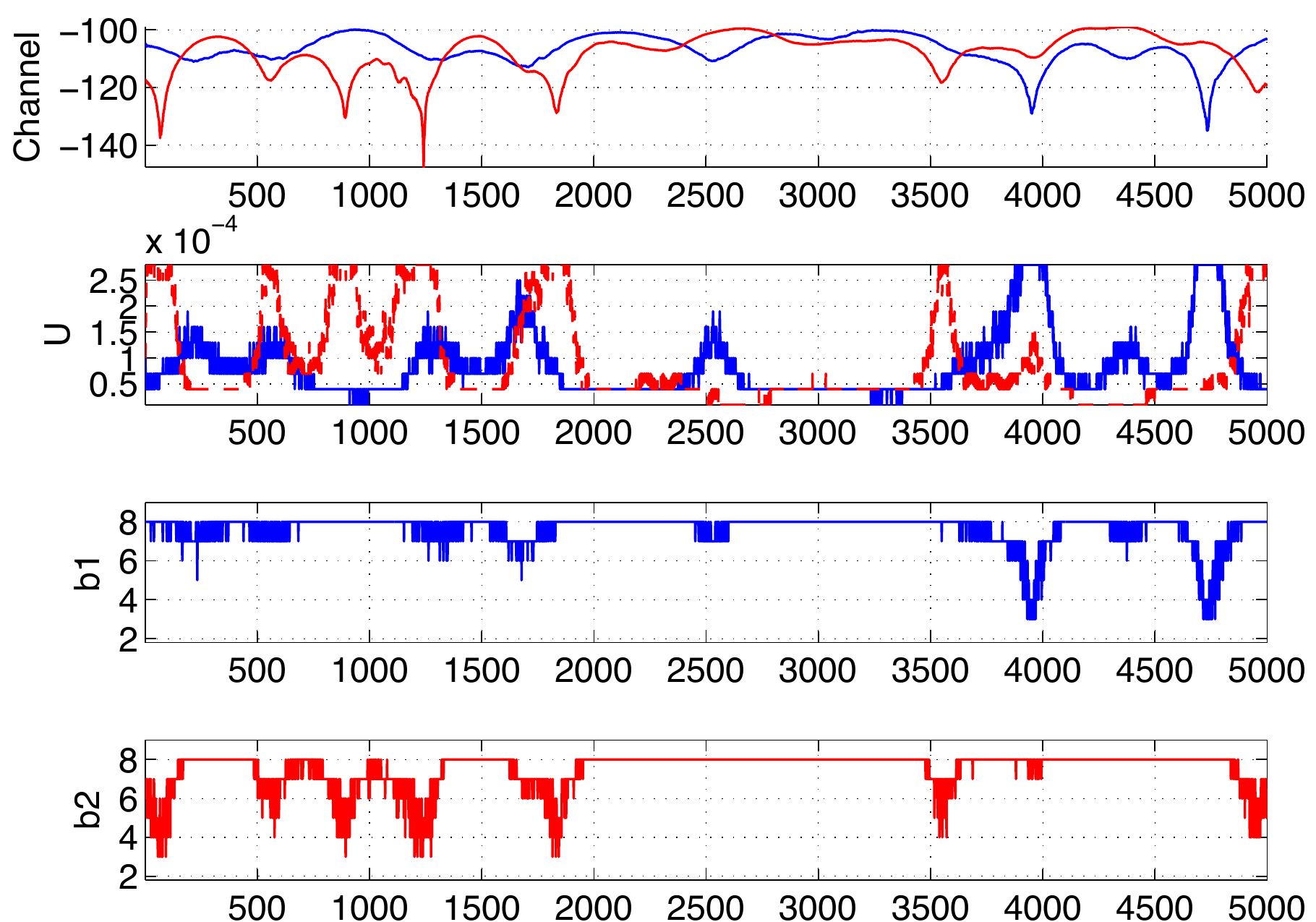}
\caption{Channel power gains: $|g_{1}(k)|^2$ (blue) and $|g_{2}(k)|^2$ (red); Predictive controller with SDC and
  constraints as in~(\ref{eq:17}).} 
\label{fig:indp_coding}
\end{figure}

\section{Results and Discussion}
\label{sec:simulations}
We consider the system described in
Section~\ref{sec:two-stage} and use measured
channel gain data obtained at the 2.4~GHz ISM band within an office
space area at the Signals and Systems group at Uppsala University,
Sweden. The top diagram of Fig.~\ref{fig:indp_coding} illustrates
the channel power gains of two realizations, one with horizontal and one
with vertical polarization. All our
  subsequent results use the same channel data,   thus facilitating performance
  comparisons.



 
\subsection{Estimation Architecture without Relays}
\label{sec:estim-arch-with}
 We first examine the performance of the  controller of 
Section~\ref{sec:predc-power-contr} for the architecture in
Fig.~\ref{fig:scheme}.

\subsubsection{Independent Coding}
A simple instance of the controller proposed is where 
only  independent coding (i.e., SDC) is allowed. Power levels and  increments
are constrained as per
\begin{equation}
\label{eq:17}
{\color{black}  0\leq u_m(k)\leq 3\times 10^{-4},\quad 
  \delta u_m(k) \in \{\pm 3\times 10^{-5} \}.}
\end{equation}
 We allow each sensor to use a
scalar 
quantizer at a bit-rate of $b_m(k)\in\{3,\dotsc,8\}$ bits/sample. The second plot in
Fig.~\ref{fig:indp_coding} shows the resulting power levels for the given channel
data. Notice that when the 
channel power gain is dropping, the 
power level is increased until reaching saturation. Then the bit-rate is decreased to
compensate for the increased expenditure; see e.g., $g_{2}(1800)$ and
$g_{1}(4750)$. The third and forth plots in 
Fig.~\ref{fig:indp_coding}, illustrate the chosen bit-rates
$b_1(k)$ and $b_2(k)$.

\begin{figure}[t]
     \centering 
\includegraphics[width=.47\textwidth]{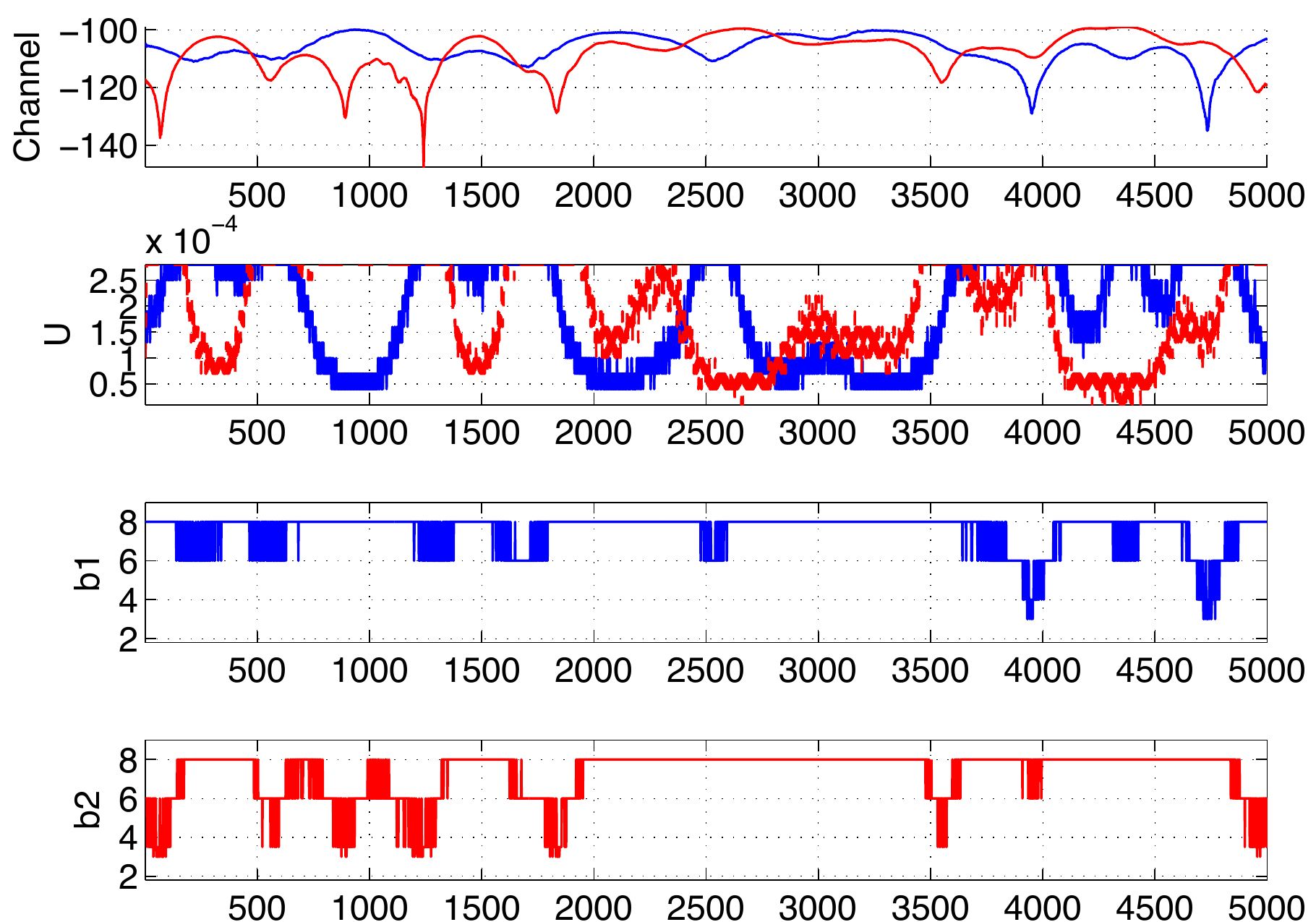}
\caption{Simple logic-based scheme.} 
\label{fig:simple}
\end{figure}

\subsubsection{Simple logic-based controller}
The \emph{independent coding} scheme described above constitutes our baseline
predictive controller. It is interesting to note that it provides a vast
improvement 
over a simpler algorithm, where the choice of the sensors transmission power
and bit-rates are based only on the predicted channel gains. In particular, let
us compare the product of the estimated channel gain $\hat{g}_m(k+1)$ and the
sensor transmission power $u_m(k)$ (from the previous time instance) to a
pre-defined threshold $T_u$. If it is above the threshold, then the power is
decreased by $\delta u_m(k)$ and if it is less than the threshold, then the power is increased by $\delta u_m(k)$. If the choice of power is above its maximum allowable power or below its minimum, it is saturated at its previous level. This simple decision logic, approximately inverts the channel. Let us define the threshold so that when the channel gain is at -110 dB, the transmission power is 0.2 mW, which gives $T_u =  10^{-110/10} \times 2\times 10^{-4}  = 2\times 10^{-15}$. 

When the channel is poor, the probability that a long bit sequence is received error-free is low. On the other hand, in good channel conditions, it makes sense to use longer bit sequences and thereby improve the estimation accuracy of the KF. With this in mind, the channel gain is split into four regions and the bits are allocated using the following rules:
\begin{equation}
b_m(k+1) =\begin{cases}
8, & \text{if $\hat{g}_m(k+1)\geq -110$ dB}, \\
6, & \text{if $-110 >\hat{g}_m(k+1) \geq -120$ dB}, \\
4, & \text{if $-120 >\hat{g}_m(k+1) \geq -130$ dB},\\
3, & \text{if $\hat{g}_m(k+1) < -130$ dB}.
\end{cases}
\end{equation}

The performance of this simple architecture is presented in
Fig.~\ref{fig:simple}. 
This simple control algorithm, when used on the given channel data, leads to an
estimation accuracy of $D=0.0637$ with a total energy usage of $98.5$ nJ. For
comparison, the baseline predictive controller  with independent coding, uses
only $45.5$ nJ at the same estimation accuracy. Thus, energy savings above $50
\%$ are possible with the predictive controllers proposed in the present
work.


\begin{figure}[t]
      \centering 
\includegraphics[width=.47\textwidth]{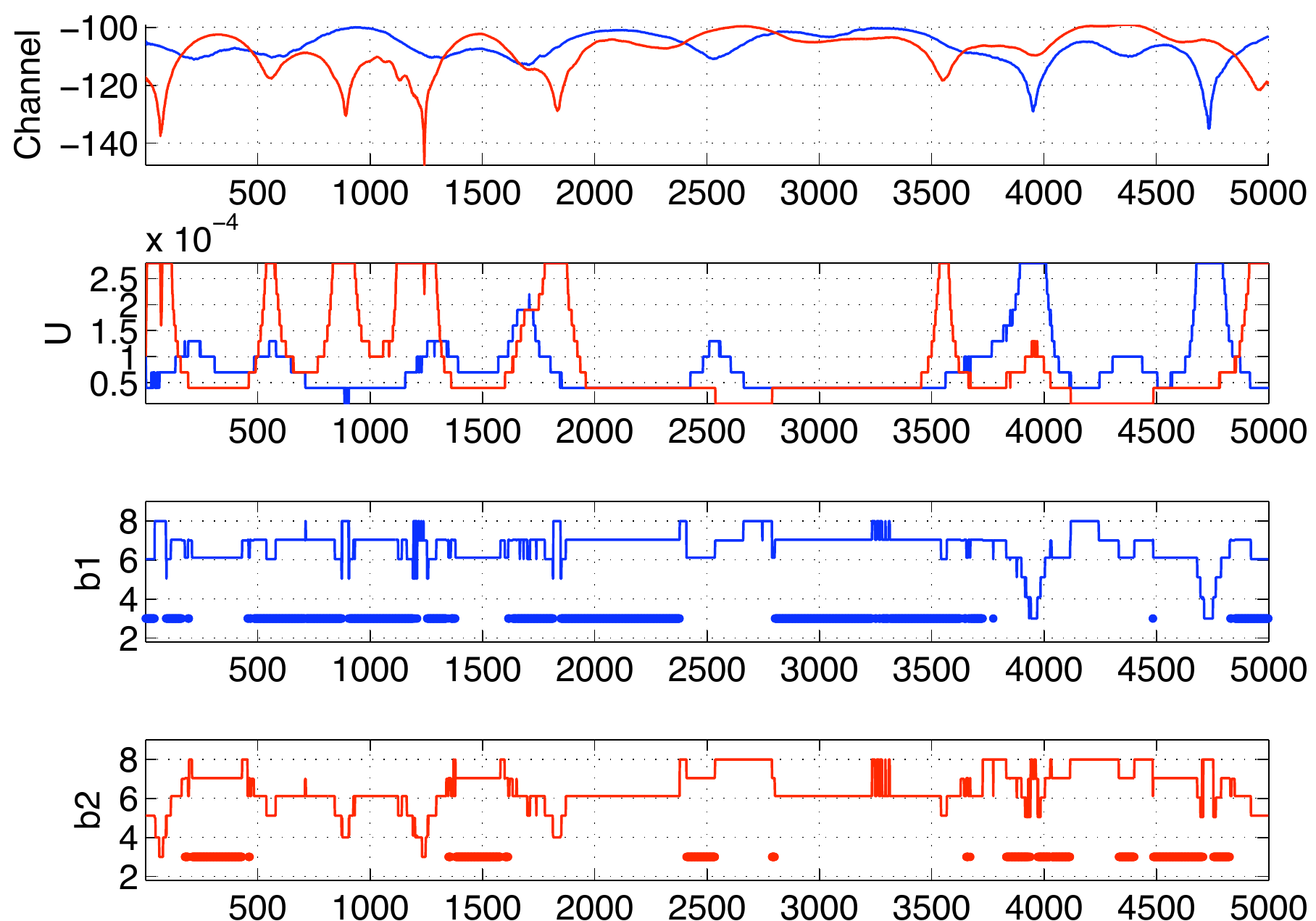}
\caption{Predictive controller with SDC and ZEC. 
  The dots indicate that ZEC is
  used and that 
  the given sensor is  dominating.}  
\label{fig:zec}
\end{figure}

\subsubsection{Robustness towards uncertainties in model dynamics}
The proposed controller relies upon knowledge of the dynamics of the underlying
process. In particular, the update of the Kalman gain requires knowledge of $A$,
and proper scaling of the quantizers require knowledge of the variances of the
processes $y_m, m=1,\dotsc,M$, which depend upon $A$. To illustrate robustness
of the scheme to model uncertainties, let us consider the
independent coding scheme described above and assume that the controller uses
$A$ given in \eqref{eq:A}, whereas the actual system dynamics is characterized by
$$\tilde{A} =
\begin{bmatrix}
  1.68& -0.99\\ 1 & 0
\end{bmatrix}.
$$
Then, the resulting variances using
$\tilde{A}$ are $\sigma_{\tilde{y}_1}^2=19.95$ and
$\sigma_{\tilde{y}_2}^2=21.93$, which are very close to those obtained using
$A$, i.e.,  $\sigma_{{y}_1}^2=21.48$ and $\sigma_{{y}_2}^2=21.93$. Indeed,  
based on simulations we have measured the average empirical entropies (bit rates) to differ by only $0.02$ bits/sample in the two cases. On the other hand, since the Kalman gain is not correctly updated in each iteration due to model mismatch, the state estimator becomes suboptimal and the estimation accuracy is reduced. Indeed, simulations show that the estimation accuracy measured by the average squared error $D=\frac{1}{K}\sum_{k=1}^{K} \|x(k)-\hat{x}(k)\|^2$ is $D=0.2056$ when $A$ is not perfectly known, as compared to $D=0.0637$ when $A$ is known. In both cases, the energy expenditure is the same ($45.5$ nJ).

\subsubsection{Independent Coding and ZEC}\label{sec:sim_zec}
Fig.~\ref{fig:zec} shows results based on the same
system and channel data as in Fig.~\ref{fig:indp_coding}; however,
this time we allow the proposed controller to use SDC as well as ZEC, see
Section~\ref{sec:zec}. In the plots showing the bit-rates, we have
used dots at the bottom to indicate when ZEC is being used. In
particular, a dot indicates that ZEC is being used across the two
sensors and that the given sensor is the dominant sensor, i.e.,
the given sensor is using independent coding whereas the other
sensor is using dependent coding. As an example, $|g_{1}(k)|^2$ is dropping at $k\approx 4000$. The controller decides upon  power
saturation for that channel and a subsequent decrease in the
bit-rate $b_1$. Here Sensor 2 is dominant, whereas Sensor 1 uses dependent coding.
\begin{figure}[t!]
      \centering 
\includegraphics[width=.47\textwidth]{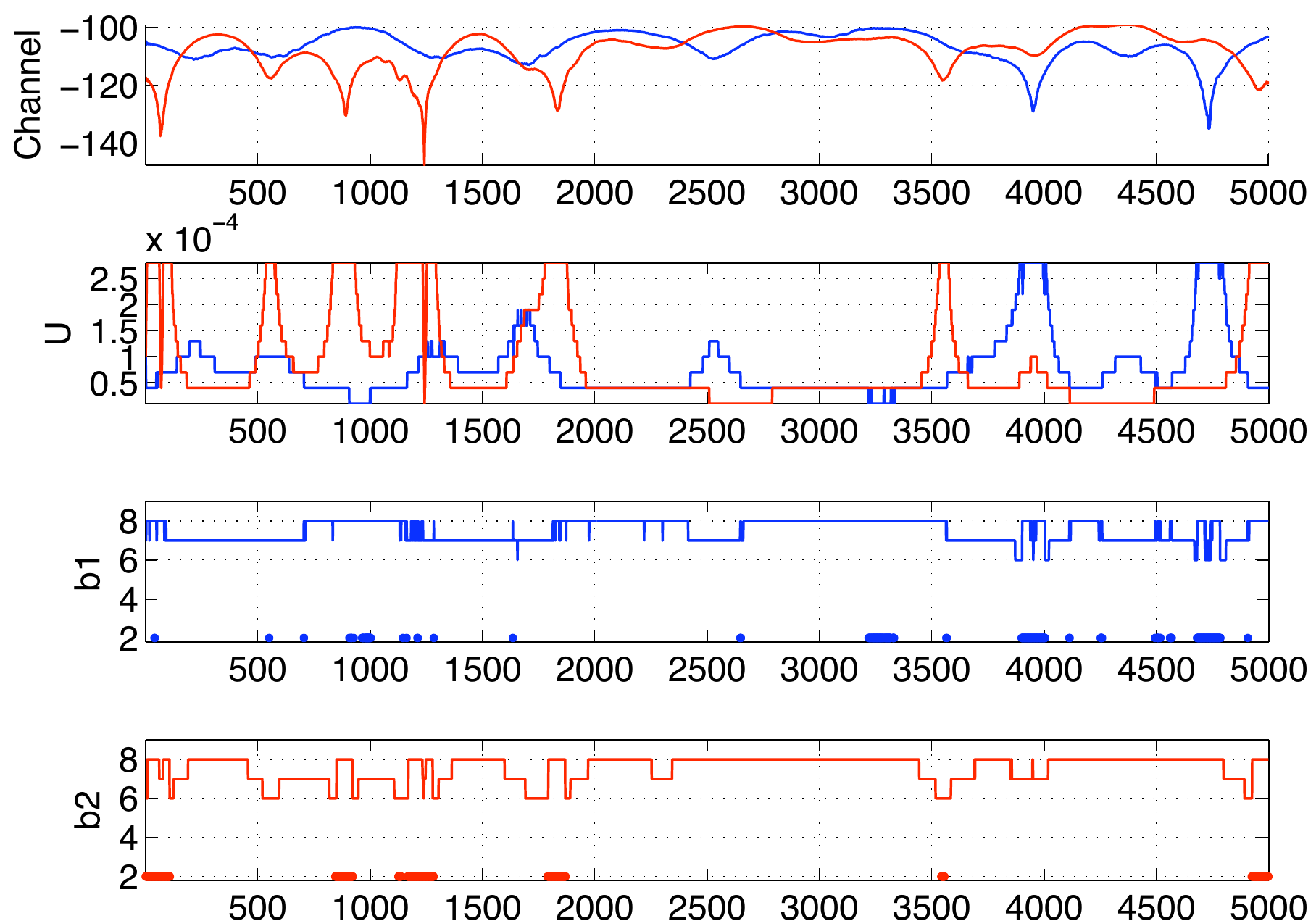}
\caption{Predictive controller with SDC and MDC. 
   The dots indicate that MDC is used.}
\label{fig:mdc}
\end{figure}

\subsubsection{Independent Coding and MDC}
We next examine a situation where the controller is allowed to use SDC and
MDC with two descriptions. Contrary to the case of ZEC, the sensors can use MDC
independently of each other.  The total bit-rate is restricted
to $b_m(k)\in\{6,7,8\}$ bits, so that the side description rates
are restricted to $\{3,3.5,4\}$ bits. The results are shown in
Fig.~\ref{fig:mdc}. The dots at the bottom of the bit-rate plots
show when MDC is  used at the particular sensor. Notice that
 especially when the channel is weak, e.g., $|g_{2}(1800)|^2$ and
$|g_{1}(4000)|^2$,  MDC is 
utilized.


\begin{figure}[t!]
       \centering 
\includegraphics[width=.47\textwidth]{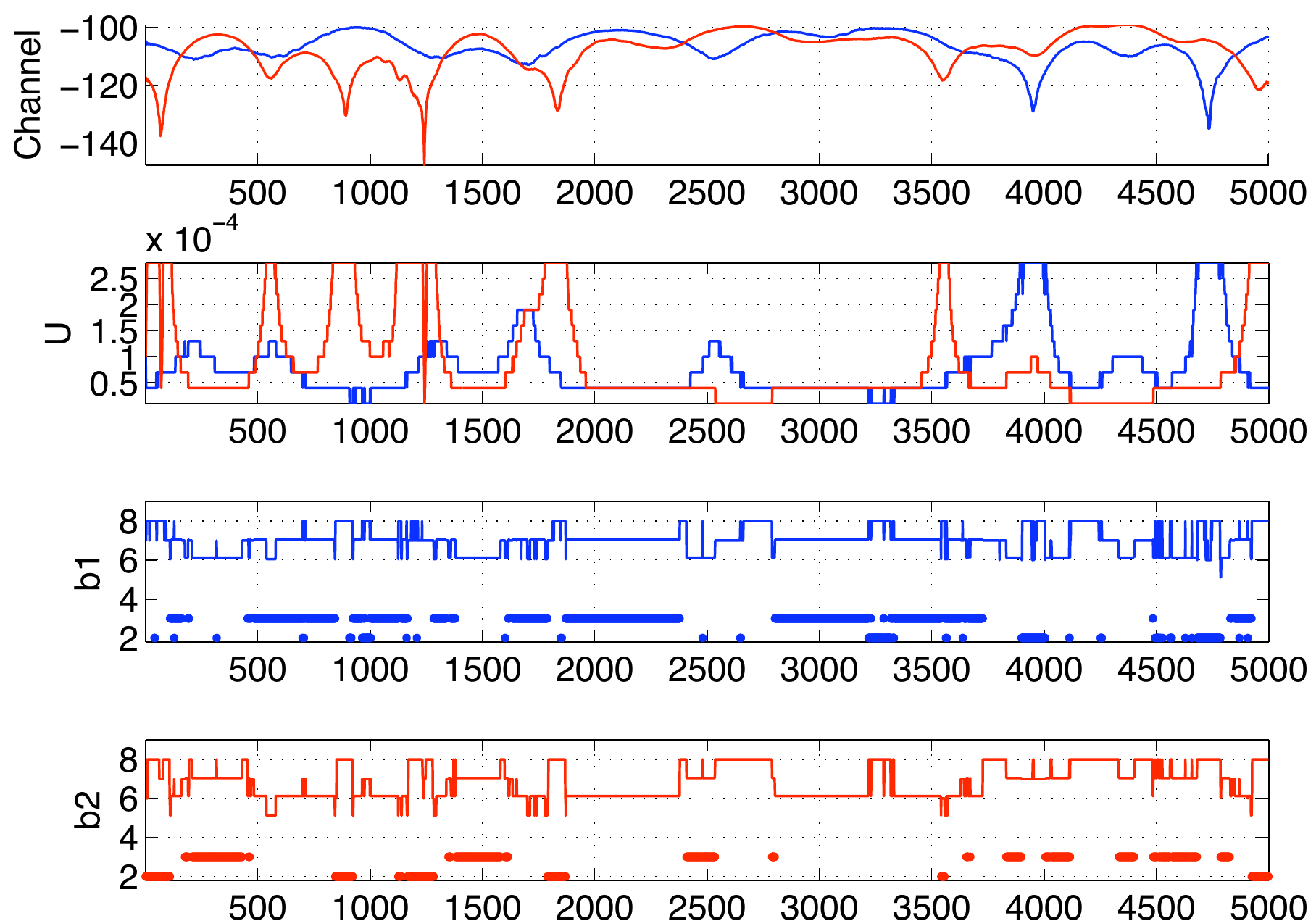}
\caption{Predictive controller with independent coding, MDC, and ZEC. 
 The dots at bit-rate level 3 indicate that ZEC is used and that the given sensor is the dominating one. The dots at bit-rate level 2 indicate that MDC is used.}
\label{fig:mdczec}
\end{figure}

\subsubsection{Independent Coding, ZEC, and MDC}
We finally combine all of the above methods and thereby allow the predictive controller to choose
SDC, ZEC, as well as MDC. We adopt the two-stage
strategy as described in Section~\ref{sec:controller}. Thus, we first evaluate
whether SDC or MDC should be used at the different
bit-rates. 
 If MDC is chosen, then the controller also finds the
optimal  trade-off between side and
central distortions. If SDC is chosen for both
sensors,  then the controller applies the second stage  involving a brute-force
search over all the bit-rates in
order to find out whether it is beneficial to use ZEC instead of
SDC. 
The results are shown in
Fig.~\ref{fig:mdczec}, where the dots at the bit-rate level 3
indicate that ZEC is used and that the given sensor is the
dominating one; the dots at bit-rate level 2 refer to
 MDC. From Fig.~\ref{fig:mdczec} we note, for example,
that for Sensor 1 ZEC is used at time $3400$ with Sensor 1 being
dominant, whereas, at time $k=4000$, MDC is used for the same
sensor.

\subsection{Estimation Accuracy vs.\ Energy Usage}
An underlying theme of our present work is the trade-off between estimation accuracy
and energy usage. In particular,  the cost function 
in~(\ref{eq:Vk}), quantifies this trade-off.
To illustrate the potential gains which can be obtained by allowing sensors to
perform coding which goes beyond independent coding, Table~\ref{tab:performance_low_power} shows
the empirical average cost, 
\begin{equation}
  \label{eq:18}
  V  \triangleq \frac{1}{5000} \sum_{k=0}^{4999} V^\star(k)
\end{equation}
where $V^\star(k)\eq\min_{S(k)} V(S(k))$,
for the different control methods examined so far.
Notice that a gain of about
1.4\% is possible simply by replacing the entropy encoders at the
sensors by entropy encoders designed based on the principle of
ZEC. It is worth emphasizing that this strategy does not rely upon an increase in the online complexity at
the sensors. 
The offline design of the entropy encoder is of course more
complicated than the design of traditional entropy encoders.

\begin{table}[t]
\begin{center}
\begin{tabular}{l|c|c}\hline
Setup & $V$ & Performance Gain \\ \hline
Indep.\ coding & 0.1256 & --- \\
Indep.\ coding + ZEC  & 0.1238 &  1.43\%    \\
Indep.\ coding + MDC  & 0.1154 &  8.12\%    \\
Indep.\ coding + ZEC + MDC & 0.1135 & 9.63\% \\ \hline
\end{tabular}
\caption{Performance of the different predictive controllers for constraints as
  in~(\ref{eq:17}).} 
\label{tab:performance_low_power}
\end{center}
\end{table}

\par With MDC   a gain of about 8.1\% is possible.
Here, it is important to recall that MDC is in our case implemented using a
single scalar quantizer followed by a table lookup, 
which maps the index of the quantized measurement to indices in side codebooks. Thus, the online complexity at the sensors
is not increased by this method either. Furthermore, the complexity at the GW is
not increased significantly, due 
to adopting the two-stage search strategy of Section~\ref{sec:controller}.
The offline design of the MDC encoders is, however, a lot more complicated than
the design of traditional 
scalar quantizers, see~\cite{ostjen06} for details.

\par Finally, if the controller is allowed to choose between independent coding,
ZEC, and MDC, then  the overall 
gain is about 9.6\%. This shows that the two more advanced coding schemes, ZEC
and MDC,  complement each other. 

%

If we increase the maximum allowable power level, then it becomes more
beneficial to use ZEC and less beneficial to use MDC. This is illustrated in
Table~\ref{tab:performance_high_power}, where we have shown the results
corresponding to the constraints
\begin{equation}
\label{eq:17b}
  0\leq u_m(k)\leq 5\times 10^{-4},\quad 
  \delta u_m(k) \in \{\pm 5\times 10^{-5} \}.
\end{equation}

\begin{figure}[t]
\begin{center}
\includegraphics[width=0.47\textwidth]{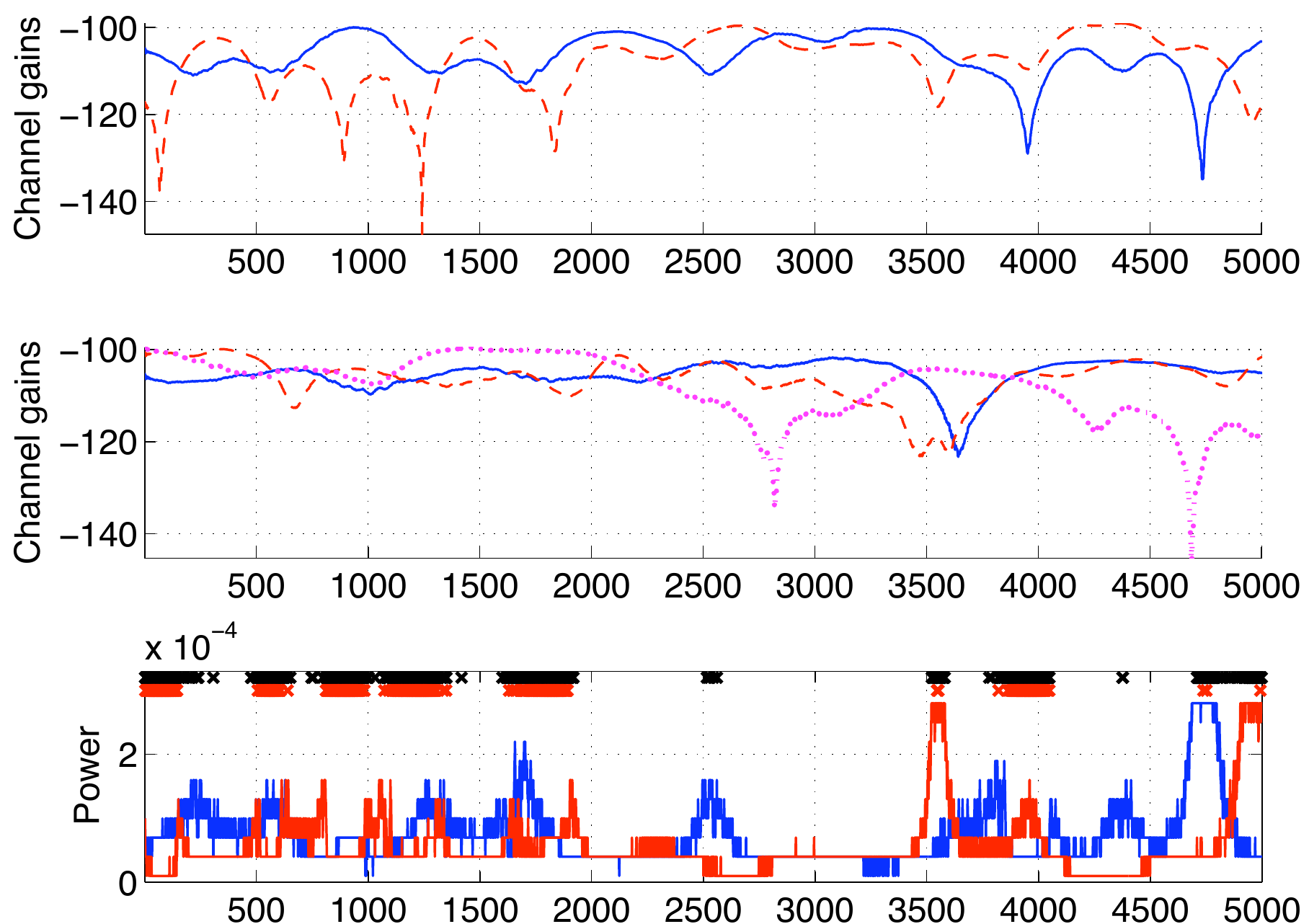}
\caption{Predictive controller for the estimation architecture in
  Fig.~\ref{fig:scheme_relay}, with known power channel gains. The GW successfully received $r_1(k)$
   in 764 out of 1291 instances where $\mu_1(k)>0$, thus, 
   the relaying
  scheme  has 59.18\% efficiency.}
\label{fig:known-channel}
\end{center}
\end{figure}

\begin{table}[t]
\begin{center}
\begin{tabular}{l|c|c}\hline
Setup & $V$ & Performance Gain \\ \hline
Indep. coding & 0.1391 & --- \\
Indep. coding + ZEC  & 0.1363 &  2.01\%    \\
Indep. coding + MDC  & 0.1326 &  4.67\%    \\
Indep. coding + ZEC + MDC & 0.1300 & 6.54\% \\ \hline
\end{tabular}
\caption{Performance of the different predictive controllers for constraints as
  in~(\ref{eq:17b}).} 
\label{tab:performance_high_power}
\end{center}
\end{table}

\subsection{Estimation Architecture with Relays}
\label{sec:estim-arch-with-1}
  We next examine the performance of the predictive controller
 when used for the estimation architecture in
  Fig.~\ref{fig:scheme_relay}. Here, the controller decides upon power level
  increments of sensors, the \emph{on-off} state of the relay, and upon bit-rates
  used by both sensors. The 
  power levels and increments of the sensors are restricted as
  per~(\ref{eq:17}). The relay performs network coding as in~(\ref{eq:4b}) with
  transmission power
  $\mu^{\text{max}}_1 = 6\times 10^{-5}$,  
  see~(\ref{eq:2}).

 \begin{table*}[t!]
\begin{center}
\fontsize{9pt}{9pt}\selectfont
\begin{tabular}{ccclcl}
{\color{black} $\varrho$} & {\color{black}$\phi$ }& ${\color{black} V_E [nJ]}$
&{\color{black} Relay Channel Models} & {\color{black}Reduction of $\phi$} 
& {\color{black}System}\\ \hline
 $10^6$ & 0.0707 & 63.21  &  -- & --  &  Baseline (no relay) \\
 $10^8$ & 0.2021 & 63.77 &  {\color{black}Sensor-Relay (predicted), Relay-GW (predicted) }& --   & Relay always on  \\
 655000 & 0.0341 & 63.11 &  {\color{black}Sensor-Relay (known), Relay-GW (known)} & 51.77\%  &  Relay on/off  \\
 680000 & 0.0382 & 63.11 &  {\color{black}Sensor-Relay (predicted), Relay-GW (predicted)} & 45.87\%  &  Relay on/off  \\
 430000 & 0.0682 & 63.12 & {\color{black} Sensor-Relay (fixed at -100 dB),  Relay-GW (predicted) }& 3.54\% &  Relay on/off  \\
 560000 & 0.0391 & 63.11 & {\color{black} Sensor-Relay (fixed at -105 dB),  Relay-GW (predicted)} & 44.70\% & Relay on/off  \\
 1003000 &  0.0375 & 63.11 &  {\color{black}Sensor-Relay (fixed at -110 dB),  Relay-GW (predicted)} & 46.96\%  &  Relay on/off  \\
 2080000 &  0.0429 & 63.11 & {\color{black} Sensor-Relay (fixed at -115 dB),  Relay-GW (predicted) }& 39.32\%  &  Relay on/off   \\ \hline
\end{tabular}
\caption{Performance gains achieved by using the relay and using network coding
  governed by the proposed controllers.} \label{tab:sim1}
\end{center}
\end{table*}

\par In Fig.~\ref{fig:known-channel}, the top diagram shows
the sensor to GW channel power gains $|g_{1}(k)|^2$ and $|g_{2}(k)|^2$ (corresponding to those used in
Section~\ref{sec:estim-arch-with}); the middle diagram show the channel power gains
 from the sensors to the relay ($|g_{1}^1(k)|^2$: blue solid line,
 $|g_{2}^1(k)|^2$: red dashed
 line) as well as the channel power gain from the relay to the GW (dotted line); the
 bottom diagram illustrates the chosen power levels of the two sensors.  
The black crosses indicate the time slots where the controller
 has decided to turn on the relay. The red crosses, on the other hand, indicate
 when the relay operation was successful, i.e., when the relay received $s_1(k)$
 and $s_2(k)$ without errors and also successfully transmitted $r_1(k)$,
 see~(\ref{eq:4b}) to the GW. It is clear
from Fig.~\ref{fig:known-channel}, 
that the controller trades off
energy spent on the sensors for energy spent on the relay. Only at
the deepest drops in $|g_{1}(k)|$ and $|g_{2}(k)|$ (occurring after $k=3500$)  the
controller chooses to saturate the sensor power levels. Note that it is
beneficial to rely on the relay and network coding most of the
time.

\par   To compare different scenarios,  we next fix  the total energy used by the
  sensors and relay  by adjusting the weighting term
  $\varrho$ in~(\ref{eq:Vk}), see Table~\ref{tab:sim1}. Therefore, the  controller seeks to distribute
the available energy between the sensor nodes and the relay to 
minimize the state estimation error variance.   As a performance
measure, we adopt the empirical value 
\begin{equation*}
  \phi \eq \frac{1}{5000} \sum_{k=0}^{4999} \tr P(k+1 |k+1).
\end{equation*}

\begin{table}[t!]
\begin{center}
\begin{tabular}{ccccc} \hline
State & Gain [dB]& $p_{k,k-1}$ & $p_{k,k}$ & $p_{k,k+1}$ \\ \hline
1 & -117.77 &  0.0000 & 0.9990 &   0.0010   \\        
2 &  -112.88&  0.0010 &  0.9978 &   0.0013  \\      
3 & -110.50&  0.0013 &   0.9973&    0.0014 \\ 
4  &  -108.83& 0.0014 &   0.9971 &   0.0015  \\
5 &-107.49&   0.0015 &  0.9970  &  0.0015  \\
6 & -106.33&   0.0015 &   0.9970  &  0.0015  \\ 
7 &-105.30&    0.0015&    0.9971 &   0.0014  \\
8 &-104.31&   0.0014 &   0.9973  &  0.0013   \\
9 & -103.32&  0.0013  &  0.9976  &  0.0011   \\
10 & -102.29&  0.0011 &   0.9981 &   0.0008   \\
 11 &-101.08&   0.0008 &   0.9986  &  0.0005 \\     
12 & -99.41&    0.0005   & 0.9995 & 0.0000 \\ \hline
\end{tabular}
\caption{State transition probabilities and the channel gains that the
  controller uses to represent the states.}
\label{tab:states}
\end{center}
\end{table}

\par  Our   baseline system uses SDC and no relay; sensor 
power levels are governed by the predictive controller. 
   In a second configuration, a  relay which is always {\em
on} is used. Since the relay uses most of the available energy
leaving very little for the sensors to spend, the performance is significantly
worse than that of the baseline 
system, see Table~\ref{tab:sim1}.
\par   Significantly better performance can be obtained if  the
  controller decides whether the relay 
shall be \emph{on} or \emph{off}. In this case  Table~\ref{tab:sim1} indicates
a  performance gain in reduction of $\phi$ of almost 52\%, provided the GW
has  exact knowledge of the sensor-relay channel gains $g_{1}^1(k)$ and
$g_{2}^1(k)$. The performance gain is almost 46\% if the GW uses simple  channel power
gain predictions.   Table~\ref{tab:sim1} also illustrates results
of situations where the relay uses only
  a constant
   estimate for the power gains of the
sensor-relay channels. Here, 
we conclude that it is safer to underestimate the sensor-relay channel power gains
than to overestimate them.

\begin{figure}[t!]
\begin{center}
\includegraphics[width=.47\textwidth]{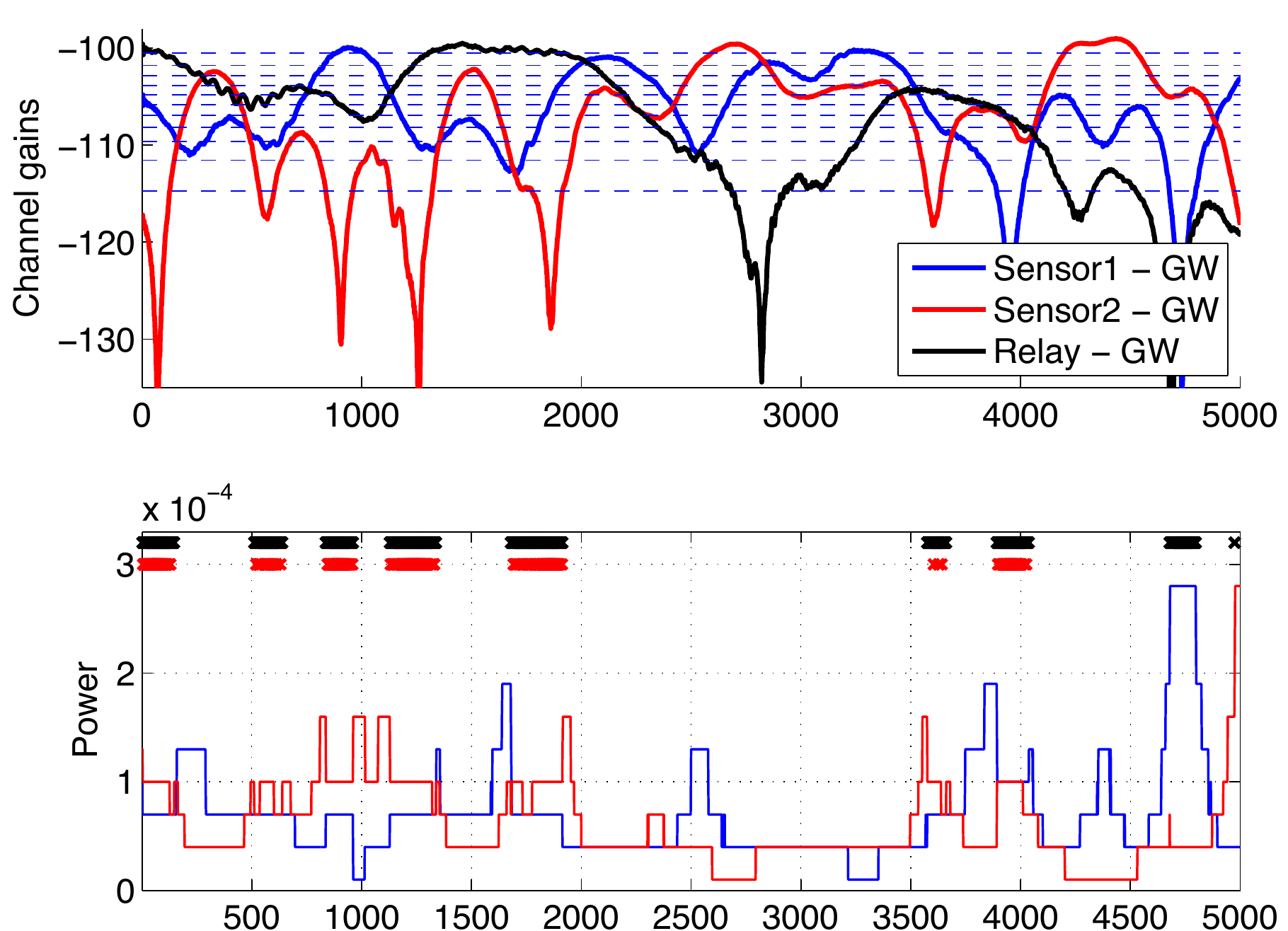}
\caption{The three channels that the GW sees are discretized into $K=12$ intervals as shown by the horizontal dashed lines. 
The GW successfully received $r_1(k)$
   in 495 out of 1196 instances where $\mu_1(k)>0$, thus, 
   giving 41.39\% efficiency.} 
\label{fig:markov_channels}
\end{center}
\end{figure}

\subsection{Estimation Architecture with Relays and Markov Channel Models}
In the previous study we assumed that the GW was able to obtain noisy
predictions of the future channel gains between the sensors and the GW as well
as the relay and the GW.  
In the following, the GW models the instantaneous fading gains of these
three channels by FSMCs, where the possible fading gains are discretized into
$K=12$ intervals (states), see Section~\ref{sec:transmission-effects} for
details. Using the approach described in~\cite{wanmoa95}, the resulting state
transition probabilities are shown in Table~\ref{tab:states}. The intervals
$\Gamma_k$ associated with the states are illustrated by horizontal dashed lines
in Fig.~\ref{fig:markov_channels} and the corresponding state channel gains are
shown in Table~\ref{tab:states}.

\begin{table*}[t!]
\begin{center}
\fontsize{9pt}{9pt}\selectfont
\begin{tabular}{ccclcl}
{\color{black} $\varrho$} & {\color{black}$\phi$ }& ${\color{black} V_E [nJ]}$
&{\color{black} Relay Channel Models} & {\color{black}Reduction of $\phi$} 
& {\color{black}System}\\ \hline
 $1.12\cdot 10^6$ & 0.0658 & 63.08  &  -- & --  &  Baseline (no relay) \\
 $1.07\cdot 10^6$ &  0.0366 & 63.00 &  {\color{black}Sensor-Relay (fixed at -110 dB),  Relay-GW (modeled)} & 44.38\%  &  Relay on/off  \\ \hline
\end{tabular}
\caption{Performance gains achieved by using the relay and using network coding
  governed by the proposed controllers and by modeling the channels as
  first-order Markovian.} \label{tab:markov} 
\end{center}
\end{table*}

 The baseline system has two sensors but no relay. The proposed system  has
 access to one relay.   
The channels between the two sensors and the relays are shown in the middle plot
of
Fig.~\ref{fig:known-channel}. These channels are unknown to the GW and simply
modeled by a fixed channel gain of $-110$ dB. The three other channels, which
are connected to the GW and shown in Fig.~\ref{fig:markov_channels}, are modeled
by the FSMC approach outlined above. The resulting performance levels are shown in
Table~\ref{tab:markov} and are, to some extent, comparable to those obtained
when using
predictions, see
Table~\ref{tab:sim1}. This observation strengthens the case of using the Markov
model in practice.  

\section{Conclusions}
\label{sec:conclusion}
We have studied state estimation with wireless sensors over correlated fading
channels. Our work shows that performance gains can be obtained by the use of
different coding schemes, when governed by a predictive 
controller which also determines  power levels. Through use of a stochastic
Lyapunov function argument we have established sufficient conditions for exponential
boundedness of the covariance of the resulting state estimation error.
 Numerical
results revealed that energy savings of more than 50\% were possible, when
compared to an alternative algorithm, wherein power levels and bit-rates are
determined by simple logic which solely depends upon channel power gains and not
the estimation error covariance. 
  It is worth noting that the coding 
  schemes examined in the present work do not require significant additional on-line
  complexity, when compared to direct quantization of the measurements.
 It is also apparent that the use of relays with simple
  network coding has the potential to
  give notable estimation performance gains, with essentially no additional
  on-line complexity at the sensor and relay nodes. Future work may include the
  study of more general network topologies and also  distributed estimation architectures
where individual nodes have additional processing capabilities; see,
e.g.,\cite{chisch11}.

\appendix[Proof of Theorem~\ref{thm:perf}]
\label{sec:proof-theor-refthm:s}
We shall consider the more general system when the system matrix $A$ is
unstable, and then evaluate the
expressions for stable $A$. To proceed, we adopt a stochastic
Lyapunov function 
approach, as presented, e.g., in\cite{kushne71,meyn89} and first
prepare the following result:

\begin{lem}
  The process $\{Z(k)\}_{k\in\N_0}$, where $Z(k)= \big(P(k|k-1),g(k-1)\big)$,
  is a Markov chain.
\end{lem}
\begin{proof}
With the model in Section~\ref{sec:transmission-effects},
$\{g(k)\}_{k\in\N_0}$ are Markovian and 
\begin{equation}
  \label{eq:11}
  \begin{split}
    \Prob&\{g(k)\,|\,Z(k),Z(k-1),Z(k-2),\dots\} \\
    &= \Prob\{g(k)\,|\,g(k-1)\} = \Prob\{g(k)\,|\,Z(k)\}.
  \end{split}
\end{equation}
 On the other hand, when using the controller of Section~\ref{sec:predc-power-contr},
  the power levels, bit-rates and coding method used at time $k$ 
  depend only on $P(k|k-1)$ and $g(k-1)$ (and deterministic
  quantities). Thus,~(\ref{eq:5}) and~(\ref{eq:8}) give that 
  the distribution of the term $K(k)C(k)$ used in~(\ref{eq:KF}) satisfies
 $   \Prob\{K(k)C(k)\,|\,Z(k),Z(k-1),\dots\} =
    \Prob\{K(k)C(k)\,|\,Z(k)\}$,
which implies that
\begin{equation*}
  \begin{split}
    \Prob&\{P(k+1|k)\,|\,Z(k),Z(k-1),Z(k-2),\dots\}\\ &= \Prob\{P(k+1|k)\,|\,Z(k)\}.
  \end{split}
\end{equation*}
Use of~(\ref{eq:11}) shows 
$\Prob\{Z(k+1)\,|\,Z(k),Z(k-1),Z(k-2),\dots\} = \Prob\{Z(k+1)\,|\,Z(k)\}$.
\end{proof}

Having established that $\{Z(k)\}_{k\in\N_0}$ is Markovian, we now adopt  the
procedure used to prove  Theorem 1 in \cite{queahl13a}   and introduce 
$ V_k\eq\tr P(k|k-1)\geq 0$.

\begin{lem}
  \label{lem:drift}
Consider $\nu(P,g)$ and $\varpi$ and $c$ defined in~(\ref{eq:30}). Then,
  \begin{equation*}
    \begin{split}
      \E&\big\{ V_{1} \,|\, Z(0)=(P,g) \big\} \leq \tr Q+(1-\nu(P,g)) \varpi
      c\\ 
      &\quad+
      \nu(P,g) \|A\|^2 \tr P ,\quad \forall (P,g)\in\R^{n\times n} \times
      \Omega.
    \end{split}
  \end{equation*}
\end{lem}
\begin{proof}
We use the total probability
formula to write:
\begin{equation}
  \label{eq:20}
  \begin{split}
     \E&\big\{ V_{1} \,|\, Z(0) =(P,g)\big\} \\
     &=  \E\big\{ V_{1} \,|\, Z(0)=(P,g),\eta(0)=0 \big\} \nu(P,g)\\
     &\quad+\E\big\{ V_{1} \,|\, Z(0)=(P,g),\eta(0)=1 \big\} (1-\nu(P,g))
     \end{split}
\end{equation}
Following as in the proof of \cite[Lemma 2]{queahl13a},  for $\eta(0)=0$ we
consider the worst case, where $\theta_m(0)=0$, 
for all $m\in\{1,\dots,M\}$. This gives:
\begin{equation}
  \label{eq:23}
   \E\big\{ V_{1} \,|\, Z(0)=(P,g),\eta(0)=0 \big\}\leq \|A\|^2\tr P + \tr Q.
\end{equation} 
To study the case where $\eta(0)=1$, we 
consider 
the simple state
predictor $\bar{x}(k+1) = AC^\dagger(k) y(k)$, where $C^\dagger(k)
\eq (C(k)^TC(k))^{-1}C(k)^T$ is the pseudo-inverse of $C(k)$.
This estimator  yields the
estimation error $x(k+1) -
\bar{x}(k+1) = w(k) + AC^\dagger(k) v(k)$, thus,
\begin{equation*}
\begin {split}
\tr\bar{P}&(k+1|k) \eq \tr\E\{(x(k+1) -
\bar{x}(k+1))\\
&\qquad\qquad \qquad\qquad \times(x(k+1) -
\bar{x}(k+1))^T\}\\
&= \tr Q + \tr \big(
AC^\dagger(k)R(k)(C^\dagger(k))^{T}A^T\big)\\
&\leq \tr Q + ||A||^2 \tr
 \big(C^\dagger(k)R(k)(C^\dagger(k))^{T}\big)\\
&\leq \tr Q + ||A||^2  \|  (C^\dagger(k))^{T}C^\dagger(k)\| \tr
R(k)
\leq \varpi c+\tr Q,
\end{split}
\end{equation*}
where we have used \cite[Fact~5.12.7]{bernst09}. Since
the Kalman filter gives the minimum conditional estimation error among all
linear estimators  we obtain that 
\begin{equation}
  \label{eq:22}
  \E\big\{ V_{1} \,|\, Z(0)=(P,g),\eta(0)=1 \big\}\leq \varpi c+\tr Q
\end{equation}
The result follows by substitution of~(\ref{eq:22}) and~(\ref{eq:23})
into~(\ref{eq:20}). 
\end{proof}
To prove Theorem~\ref{thm:perf}, we use $V_k$  as a
candidate Lyapunov function. Lemma~\ref{lem:drift} and~(\ref{eq:13}) give that
\begin{equation*}
  \begin{split}
    0&\leq \E\big\{ V_{1} \,|\, Z(0)=(P,g) \big\} \leq \tr Q+(1-\nu(P,g)) \varpi
    c\\
    &+ \nu(P,g) 
  \|A\|^2 \tr P \leq\nu(P,g) \|A\|^2 V_0 +\bar{\beta} \leq \rho V_0
+\bar{\beta},
  \end{split}
\end{equation*}
 for all $(P,g)$, and where 
$\bar \beta \eq \tr Q+(1-\nu(P,g)) \varpi c \leq \tr Q+ \varpi c\in[0,\infty)$. Since
$\{Z(k)\}_{k\in\N_0}$ is Markovian, we can use \cite[Prop.\ 3.2]{meyn89} (see
also\cite{queahl13a}) to
conclude that
\begin{equation}
    \label{eq:39b}
    0\leq \E \big\{ V_k   \,|\, Z(0)=Z \big\} \leq \rho^k V_0 + \bar{\beta}
    \sum_{i=0}^{k-1}\rho^i= \rho^k V_0 +\bar{\beta} \frac{1-\rho^k}{1-\rho},
  \end{equation}
 for all $k\in\N_0$.  Since $P(k|k-1)\succeq 0$,
  it holds that $ V_k\geq \|P(k|k-1)\|$, for all $ k\in
      \N_0$. Thus, upon noting that $P(0|-1)=P_0$ is
  given,~\eqref{eq:39b} establishes~(\ref{eq:40}).
\hfs

\bibliography{/Users/daniel/Dropbox/dquevedo}

\end{document}